\documentclass[11pt]{article}
\usepackage{fullpage}

\title{Fast Distributed Algorithms for Girth, Cycles and Small Subgraphs} 

\author{Keren Censor-Hillel \thanks{Technion, Israel Institute of Technology, Israel. Email:ckeren@cs.technion.ac.il} ~ Orr Fischer \thanks{Tel-Aviv University, Israel. Email: orrfischer@mail.tau.ac.il} ~ Tzlil Gonen \thanks{Tel-Aviv University, Israel. Email: tzlilgon@gmail.com} \\ Fran{\c c}ois Le Gall \thanks{Nagoya University, Japan. Email: legall@math.nagoya-u.ac.jp} ~ Dean Leitersdorf \thanks{Technion, Israel Institute of Technology, Israel. Email: dean.leitersdorf@gmail.com} ~ Rotem Oshman \thanks{Tel-Aviv University, Israel. Email: roshman@tau.ac.il}}
\date{}

\usepackage{amsmath,amssymb,amsthm}
\usepackage{color}
\usepackage{graphicx}
\usepackage{enumerate}

\usepackage[hidelinks]{hyperref}
\usepackage{cite}
\usepackage{caption}
\usepackage{subcaption}
\usepackage{bm}

\makeatletter
\newtheorem*{rep@theorem}{\rep@title}
\newcommand{\newreptheorem}[2]{%
	\newenvironment{rep#1}[1]{%
		\def\rep@title{#2 \ref{##1}}%
		\begin{rep@theorem}}%
		{\end{rep@theorem}}}
\makeatother

\newenvironment{lemma-repeat}[1]{\begin{trivlist}
		\item[\hspace{\labelsep}{\bf\noindent Lemma \ref{#1}.}]\em }%
	{\end{trivlist}}
\newenvironment{theorem-repeat}[1]{\begin{trivlist}
		\item[\hspace{\labelsep}{\bf\noindent Theorem \ref{#1}.}]\em }%
	{\end{trivlist}}
\newenvironment{claim-repeat}[1]{\begin{trivlist}
		\item[\hspace{\labelsep}{\bf\noindent Claim \ref{#1}.}]\em }%
	{\end{trivlist}}
	\newenvironment{proposition-repeat}[1]{\begin{trivlist}
		\item[\hspace{\labelsep}{\bf\noindent Proposition \ref{#1}.}]\em }%
	{\end{trivlist}}

\newcommand{\qedsymb}{\qed}
\newenvironment{proofof}[1]{\begin{trivlist}
		\item[\hspace{\labelsep}{\bf\noindent Proof of #1: }]
	}{\qedsymb\end{trivlist}}

\newtheorem{theorem}{Theorem}
\newtheorem{claim}{Claim}
\newtheorem{lemma}{Lemma}
\newtheorem{observation}{Observation}
\newtheorem{proposition}{Proposition}
\newtheorem{corollary}[theorem]{Corollary}

\newtheorem{definition}{Definition}

\newcommand{\set}[1]{\left\{ #1 \right\}}

\DeclareMathOperator{\polylog}{\mathrm{polylog}}

\newcommand{\remove}[1]{}

\usepackage{xcolor}

\usepackage{amsmath}
\usepackage{amsfonts}
\usepackage[linesnumbered,vlined]{algorithm2e}
\usepackage{comment}
\usepackage{amsthm}
\graphicspath{ {./images/} }

\usepackage{latexsym}
\usepackage{amssymb}
\usepackage{todonotes}
\usepackage{comment}
\usepackage{mathtools}
\usepackage{tikz}

\theoremstyle{remark}

\theoremstyle{definition}

\newcommand{\ip}[1]{\left}
\usepackage{xspace}
\newcommand{\congest}{\textsc{Congest}\xspace}

\newcommand{\clique}{\textsc{Congested Clique}\xspace}

\newcommand{\parent}{\texttt{parent}}
\DeclareMathOperator*{\ex}{ex}

\newcommand{\hide}[1]{ }

\newcommand{\nat}{\mathbb{N}}

\renewcommand{\mathbf}{\bm}
\newcommand{\var}[1]{\mathit{#1}}
\newcommand{\stt}{\medspace \mid \medspace}

\RestyleAlgo{boxruled}
\LinesNumbered
\let\oldnl\nl
\newcommand{\nonl}{\renewcommand{\nl}{\let\nl\oldnl}}

\SetKw{Send}{send}

\begin{document}
		\title{
			Fast Distributed Algorithms for Girth, Cycles and Small Subgraphs 
		}
%
		
		\maketitle
\begin{abstract}

In this paper we give fast distributed graph algorithms for detecting and listing small subgraphs, and for computing or approximating the girth. Our algorithms improve upon the state of the art by polynomial factors, and for girth, we obtain a constant-time algorithm for additive +1 approximation in \clique, and the first parametrized algorithm for exact computation in \congest.

In the \clique model, we first develop a technique for learning small neighborhoods, and apply it to obtain an $O(1)$-round algorithm that computes the girth with only an additive $+1$ error. Next, we introduce a new technique (the partition tree technique) allowing for efficiently listing all copies of any subgraph, which is deterministic and improves upon the state-of the-art for non-dense graphs. We give two concrete applications of the partition tree technique: First we show that for constant $k$, it is possible to solve $C_{2k}$-detection in $O(1)$ rounds in the \clique, improving on prior work, which used fast matrix multiplication and thus had polynomial round complexity. Second, we show that in triangle-free graphs, the girth can be exactly computed in time polynomially faster than the best known bounds for general graphs. We remark that no analogous result is currently known for sequential algorithms.

In the \congest model, we describe a new approach for finding cycles, and instantiate it in two ways: first, we show a fast parametrized algorithm for girth with round complexity $\tilde{O}(\min\{g\cdot n^{1-1/\Theta(g)},n\})$ for any girth $g$; and second, we show how to find small even-length cycles $C_{2k}$ for $k = 3,4,5$ in $O(n^{1-1/k})$ rounds. This is a polynomial improvement upon the previous running times; for example, our $C_6$-detection algorithm runs in $O(n^{2/3})$ rounds,
compared to $O(n^{3/4})$ in prior work.
Finally, using our improved $C_6$-freeness algorithm, and the barrier on proving lower bounds on triangle-freeness of Eden et al., we show that improving the 
current $\tilde\Omega(\sqrt{n})$ lower bound for $C_6$-freeness of Korhonen et al.~by \emph{any} polynomial factor would imply strong circuit complexity lower bounds.		\end{abstract}
		
\newpage
\section{Introduction}
\label{sec:intro}

A fundamental problem in many computational settings is that of finding cycles and other small subgraphs within a given graph. This paper focuses on finding subgraphs in distributed networks that communicate through limited bandwidth. The motivation for this is two-fold: first, for some subgraphs $H$ there exist distributed algorithms that perform better on $H$-free graphs, such as distributed cut and coloring algorithms in triangle-free graphs~\cite{HRSS17, PS15}. 
%
The second reason for which we are interested in these problems is that while solving them only requires obtaining \emph{local} knowledge, about small non-distant neighborhoods, the bandwidth restrictions impose a major hurdle for collecting this information. This induces a rich landscape of complexities for subgraph-related problems. We contribute to the effort of characterizing the complexities of subgraph-related problems by providing new techniques, from which we derive fast algorithms for such problems in the two key distributed bandwidth restricted models, namely, \congest and \clique.

In the \clique model, $n$ synchronous nodes can send messages of $O(\log n)$ bits in an all-to-all fashion. The input graph is an arbitrary $n$ vertex graph, partitioned such that every node receives the edges of a single vertex as input. Our main contribution in this model is an algorithm for obtaining a $+1$ approximation for the girth in a \emph{constant} number of rounds, where the girth of a graph is the length of its shortest cycle. 
\newcommand{\TheoremGirthApprox}
{
	Given a graph $G$ with an unknown girth $g$, there exists a deterministic $O(1)$ round algorithm in the \clique model which outputs an integer $a$, such that $g \in \{a, a+1\}$.
}
\begin{theorem}
	\label{theorem:girthApprox}
	\TheoremGirthApprox
\end{theorem}
For comparison, note that the current state-of-the-art algorithm computes the exact girth in $O(n^{0.158})$ rounds~\cite{CHKKLPJ17}.
To obtain our $+1$ approximation algorithm, we devise two main new methods, which we describe here in a nutshell. The first is an algorithm in which each node learns its entire neighborhood up to a radius which is a constant approximation of the girth. To this end, we prove that we can quickly list all paths of sufficient length, as well as efficiently distribute them to the nodes that need to learn them. The second method that we introduce is a way to double the radius of the neighborhoods that all the nodes know, by having each node acquire the knowledge held by the farthest nodes in its currently-known neighborhood. Crucially, both of these procedures can be done in $O(1)$ rounds, and could be useful for additional applications.

Our second contribution in the \clique model is a \emph{partition tree} technique which allows for efficiently detecting or listing all copies of any subgraph with at most $\log n$ nodes, in a deterministic manner. In particular, our main application of the partition tree technique is to obtain the following subgraph listing algorithm, which improves upon the state-of-the-art for non-dense graphs.


\newcommand{\TheoremSubgraphsCC}{
	Given a graph $G$ with $n$ nodes and $m$ edges and a graph $H$ with $p \leq \log n$ nodes and $k$ edges, let $\tilde{m}=\max\{m,n^{1+1/p}\}$. 	
	There exists a deterministic \clique algorithm that terminates in $O(\frac{k\tilde{m}}{n^{1+2/p}} + p)$ rounds and lists all instances of $H$ in $G$.
}
\begin{theorem}
	\label{theorem:generalSubgraphTheorem}
	\TheoremSubgraphsCC
\end{theorem}

We give two concrete applications of this result. The first is fast detection of even cycles.
\newcommand{\TheoremCyclesCC}
{
	Given a graph $G$ and an integer $k \leq (\log n)/2$, there exists a deterministic $O(k^2)$-round algorithm in the \clique model for detecting cycles of length $2k$.
}
\begin{corollary}
	\label{corollary:CyclesCC}
	\TheoremCyclesCC
\end{corollary}
Note that for constant $k$ the above algorithm completes within $O(1)$ rounds. Prior work for cycle detection in the \clique model used fast matrix multiplication (FMM) and thus had polynomial round complexity, apart from detecting $4$-cycles which was shown to have a constant-round algorithm~\cite{CHKKLPJ17}.
The second implication of the partition-tree technique is a fast algorithm for computing the \emph{exact} girth in triangle-free graphs. Prior algorithms for girth in the \clique model are based on fast matrix multiplication (FMM), a technique that can be no faster than checking for triangle-freeness.

\begin{corollary}
	\label{corollary:triangle-free-graphs}
	Given a triangle-free graph $G$ with an unknown girth $g$, there exists a deterministic $\tilde O(n^{1/10})$-round algorithm in the \clique model which outputs $g$.
\end{corollary}
This result leverages the fact that graphs without small cycles become increasingly sparse, and the algorithm of Theorem \ref{theorem:generalSubgraphTheorem} is efficient on sparse graphs. 
We remark that, interestingly, no analogous result going below the complexity of FMM for girth in triangle-free graphs is known for sequential algorithms, since the best known sequential algorithms for cycle detection in sparse graphs (see \cite{AYZ97}) are not fast enough.
We also note that given further lower bounds on the girth beyond triangle-freeness, the runtime of our algorithm improves even further; for instance, if the graph does not contain any $k$-cycle for $k\in\{3,4,5\}$, then our algorithm computes the exact girth in $\tilde O(n^{1/21})$ rounds. 
We refer to Proposition~\ref{proposition:exactGirth} in Section \ref{section:listing} for a more precise statement.

In the \congest model, $n$ synchronous nodes can send messages of $O(\log n)$ bits to their neighbors only, and the input graph is the communication graph. In this model, we develop a new approach for finding cycles of a given size. A key step that is present in all known sublinear-round algorithms for finding cycles in \congest is the elimination of \emph{high-degree vertices}: we check whether there is a cycle that includes a high-degree node, and if we conclude that there is no such cycle, we can remove the high-degree nodes from the graph. The remaining graph is much easier to handle, since it has low degree.
In prior work, the high-degree vertices were eliminated by sequentially enumerating over them and starting a short BFS from each one.
Here we introduce a different method for finding cycles that include a high-degree node:
intuitively, we show that if we start from a \emph{neighbor} of a small even cycle, we can quickly find the cycle itself.
Since high-degree nodes have many neighbors, if we sample a uniformly random node in the graph,
we are somewhat likely to hit a neighbor of the high-degree node, and from there we can find the cycle in constant rounds.

We apply this technique to give a fast algorithm that detects small even cycles, and a fast parameterized algorithm for computing the \emph{exact} girth. Specifically, we obtain the following:

\newcommand{\TheoremCyclesCongest}{
	Given a graph $G$, there exists a randomized algorithm in the \congest model for detection of $2k$-cycles in $O(n^{1-1/k})$ rounds, for $k=3,4,5$.
}
\begin{theorem}
	\label{theorem:cyclesCongest}
	\TheoremCyclesCongest
\end{theorem}
This significantly improves upon the running time of $O(n^{1 - 1/\Theta(k^2)})$ of the previous state-of-the-art \cite{EFFKO19}:
for cycles of length 6,8 or 10, the previous algorithm had running time $O(n^{3/4})$, $O(n^{5/6})$ or $O(n^{10/11})$, respectively.
We believe that going below round complexity of $O(n^{1-1/k})$ for $C_{2k}$-detection in the \congest model would require a breakthrough beyond currently known techniques, with potential ramifications also for the \clique model. 

For exact girth, previously, an $O(n)$-round algorithm for exact girth was known, based on computing all-pairs shortest paths~\cite{HW12}. Our result is as follows:
\newcommand{\TheoremExactGirthCongest}{
	Given a graph $G$ with an unknown girth $g$, there exists a randomized $O(\min\{g\cdot n^{1-1/\Theta(g)},n\})$-round algorithm in the \congest model which outputs $g$.
}
\begin{theorem}
	\label{theorem:exactGirthCongest}
	\TheoremExactGirthCongest
\end{theorem}
``Outputs'' here means that the first node that halts outputs the girth. Other nodes of the graph may halt later, and output larger values. This is unavoidable, unless we introduce a term in the running time that depends on the diameter of the graph.

Our final result is an \emph{obstacle} on proving lower bounds for $C_6$-freeness in \congest.
In~\cite{KR17} it was shown that the $C_{2k}$-freeness problem is subject to a lower bound of $\widetilde{\Omega}(\sqrt{n})$, for any $k$.
For $C_6$-freeness, the best known algorithm is our new algorithm here, which runs in $\tilde{O}(n^{2/3})$ rounds,
and there are reasons to believe that this may be optimal.
Unfortunately, we show that proving a lower bound of the form $\Omega(n^{1/2+\alpha})$, for any constant $\alpha > 0$,
would imply breakthrough results in circuit complexity.
This result uses ideas from our improved $C_6$-freeness algorithm, and the barrier on proving lower bounds on triangle-freeness from~\cite{EFFKO19}.
%
%
\subparagraph*{Related work.}
\label{subsec:related}
The problem of subgraph-freeness, and in particular cycle detection, has been extensively studied in the \clique and \congest models. While there are only a few papers which study girth computation, related problems such as diameter computation or shortest paths were also extensively studied in these models.

In the first work to consider girth computation in the sequential setting, Itai and Rodeh~\cite{IR78} gave algorithms with running time $O(mn)$ and $O(n^2)$ for computing exact girth and $+1$ approximation of the girth, respectively, using a BFS approach, and an $O(n^\omega)$ algorithm for exact girth using an algebric method, where $\omega$ is the exponent of matrix multiplication. Later, various trade-offs between running time and additive or multiplicative approximations for girth were obtained (e.g \cite{LL09,RW12,RT13,RW11}).

In the \clique model, an $O(n^{0.158})$ round algorithm for exact girth and a $2^{O(k)}n^{0.158}$ round algorithm for $C_k$-detection for any $k$ was shown by \cite{CHKKLPJ17} based on matrix multiplication techniques. These algebraic techniques were later extended by~\cite{LG16,CHLT20,CHDK19}.

For general subgraphs, a \clique algorithm for listing all instances of a subgraph $H$ of size $k$ in $\widetilde{O}(n^{1-2/k})$ rounds was shown in \cite{DLP12}, which was shown to be tight for triangles~\cite{PS18,ILG2017} and later for cliques of size $k>3$ as well \cite{FGKO18}. In this work we give a ``sparsity-aware'' version of this result, which has improved performance as the graph becomes sparser. Previously, distributed ``sparsity-aware'' algorithm were studied in the context of distributed sparse matrix multiplication~\cite{CHLT20,CHDK19} and in the context of the $k$-machine model~\cite{PS18}.

Frischknecht et al.~\cite{FHW12} was the first work to consider girth computation in \congest, and showed that at least $\widetilde{\Omega}(\sqrt{n})$ rounds are required in order to obtain a $(2-\epsilon)$-approximation of the girth. Peleg et al. \cite{PRT12} showed an algorithm computing a $(2-1/g)$-approximation of the girth with round complexity $O(D+\sqrt{gn}\log{n})$, where $g$ is the girth of the graph. Holzer et al. \cite{HW12} showed an algorithm for exact girth computation in $O(n)$ rounds, based on an exact all-pairs shortest path algorithm, and an algorithm for computing an $(1+\epsilon)$-approximation of the girth in $O(\min\{n/g+D\log{(D/g)},n\})$ rounds.

In the cycle-freeness problem in the \congest setting, Drucker et al.~\cite{DKO14} showed a near tight lower bound of $\widetilde{\Omega}(n)$ for constant sized odd-length cycles, as well as a lower bound of $\widetilde{\Omega}(n^{1/k})$ for $C_{2k}$-freeness, which was later improved to $\widetilde{\Omega}(\sqrt{n})$ by Korhonen et al.~\cite{KR17}. A \congest randomized algorithm for listing all triangles with round complexity $\widetilde{O}(n^{1/3})$ was shown in~\cite{CS19}, which improved the previous $\widetilde{O}(n^{1/2})$-round algorithm of~\cite{CPZ19} and the $\widetilde{O}(n^{3/4})$-round algorithm of~\cite{ILG2017}. The first sublinear-time algorithm for $C_{2k}$-freeness for $k \geq 3$ was given in~\cite{FGKO18} running in $\widetilde{O}(n^{1-1/(k^2-k)})$ rounds, and was later improved in~\cite{EFFKO19} to round complexity $\widetilde{O}(n^{1-2/(k^2-k+2)})$ for odd $k$ and $\widetilde{O}(n^{1-2/(k^2-2k+4)})$ for even $k$.


\section{Preliminaries}
\subparagraph*{Definitions.}
 Given a graph $H$, the \emph{$H$-listing} problem is a problem in which each node may output a set of $H$-copies, and the goal of the network is that w.h.p. the union over the sets of outputted $H$-copies by the nodes is exactly the set of $H$-copies in $G$. 


The \emph{T{\'u}ran number} of a graph $H$, denoted $\ex(n,H)$, is the maximum number of edges $m$ such that there exists a graph $G$ on $n$ vertices and $m$ edges which contains no copy of $H$. 

\begin{lemma}[T{\'u}ran number of $C_{2k}$ \cite{ZM13}]
	\label{lem:prelim_turan}
	For $k \in \mathbb{N}$, if $G$ is $C_{2k}$-free, then $G$ contains at most $17kn^{1+1/k}$ edges.
\end{lemma}

\begin{lemma}[T{\'u}ran number for girth \cite{ZM13}]
	\label{lem:prelim_girth_turan}
	If $G$ is $C_i$-free for all $3 \leq i \leq 2k$, then $G$ contains at most $n^{1+1/k} + n$ edges.
\end{lemma}

Let $N_i(v)$ denote the \emph{graph} defined by the nodes of hop-distance at most $i$ from $v$, that is, 
it includes all such nodes and all the \emph{edges} incident to nodes with hop-distance at most $i-1$ from $v$. For sets $A,B \subseteq V$, denote by $E(A,B) = \{(a,b) \in E \stt a \in A \land b \in B\}$ the set of edges between $A$ and $B$.

\subparagraph*{Load-Balanced Routing in the \clique Model.}
We introduce a useful routing procedure which extends that of \cite{Lenzen13}, and it is used throughout our results. The routing procedure of \cite{Lenzen13} routes a set of messages where each node needs to send and receive at most $O(n)$ messages, in $O(1)$ rounds. The following shows that it possible to replace the constraint where each node needs to \emph{send} at most $O(n)$ messages with one stating that the messages each node desires to send are based on at most $O(n\log n)$ bits. This allows us to route messages even when the some nodes are each a source of $\omega(n)$ messages.

\newcommand{\LoadBalancedRouting}
{
	Any routing instance $\mathcal{M}$, in which every node $v$ is the target of up to $O(n)$ messages, and $v$ locally computes the messages it desires to send from at most $|R(v)| = O(n\log{n})$ bits, can performed in $O(1)$ rounds. 
}
\begin{lemma}[Load Balanced Routing]
	\label{lem:load_balance_routing}
	\LoadBalancedRouting
\end{lemma}
\begin{proof}
	For every node $v$, let $s(v)$ be the number of messages $v$ is a source of in $\mathcal{M}$, and have $v$ broadcast $s(v)$. The network allocates $h(v) = \lceil s(v)/n \rceil$ \emph{helper nodes} to $v$ in such a way that each node $u$ is a helper node of at most $O(1)$ other nodes, and such that all nodes can locally compute which node helps another node. This is possible as since each node is the target of at most $O(n)$ messages, then the total number of messages is at most $(c-1)n^2$, for some constant $c$, and therefore $\sum_v h(v) \leq cn$. 
	
	Having allocated the helper nodes $h(v)$ to each node $v$, we ensure that these nodes learn $R(v)$ - this will later allow them to reproduce the messages which $v$ desires to send. First, each node $v$ partitions $R(v)$ into $O(n)$ messages of size $\log{n}$ and sends the $i^{th}$ message to node $i$. Then, node $i$ sends to each node $u \in h(v)$, the message which it got from $v$. Notice that since each node $u$ is the helper of at most $O(1)$ other nodes, then every node $i$ needs to send to every other node $u$ at most $O(1)$ messages. Therefore, this takes $O(1)$ rounds of communication since $O(1)$ messages are sent on every communication link $i-u$. 
	
	Every node in $h(v)$ now knows all of $R(v)$, and can thus locally create the messages $v$ desires to send. Node $v$ splits the messages it desires to send into $h(v)$ sets $(M_1(v),\dots,M_{h(v)}(v))$, each of size at most $O(n)$, and assigns each of the sets to one of its helper nodes. As the targets do not change, each node is still the target of up to $O(n)$ messages. Further, every helper node is now the source of up to $O(n)$ messages. Therefore, it is possible to apply the routing scheme from~\cite{Lenzen13} in order to have each helper node route the messages which it is assigned.
\end{proof}

Notice that this lemma implies, in a straightforward manner, the following Corollary~\ref{cor:load_balancing_routing_deterministic} which we refer to extensively.
\begin{corollary}
	\label{cor:load_balancing_routing_deterministic} 
	In the deterministic \clique model, given that each node originally begins with $O(n \log n)$ bits of input, and at most $O(1)$ rounds have passed since the initiation of the algorithm, then any routing instance $\mathcal{M}$, in which every node $v$ is the target of up to $O(n \cdot x)$ messages, can performed in $O(x)$ rounds.
\end{corollary}

While this is a weaker statement than that of Lemma~\ref{lem:load_balance_routing}, it is convenient to use when showing constant-time algorithms in the \clique model, as it completely circumvents the need for a bound on the number of messages each node desires to send.
\section{Deterministic $O(1)$ Round Algorithm for +1 Girth Approximation in the Congested Clique}
In this section we prove Theorem \ref{theorem:girthApprox}: we construct a deterministic $O(1)$ round algorithm for $+1$ girth approximation in the \clique model.
%
%
%
The algorithm is composed of two phases, each a novel technique on its own, and through their combination, we achieve the desired result. 
The first procedure is based on a \emph{subgraph enumeration} approach and allows each node to learn its $\lfloor \frac{g}{20} \rfloor$ hop-neighborhood in $O(1)$ rounds. Formally, it is shown in the following Theorem~\ref{theorem:edgePartitioning}.

\newcommand{\TheoremEdgePartitioning}
{
	Given a graph $G$, with $n$ nodes, and an unknown girth $g < \log n$, there exists an $O(1)$ round algorithm in the \clique model, which either outputs $g$, or, ensures that every node knows its $\lfloor  g / 20  \rfloor$ hop-neighborhood.
}
\begin{theorem}
	\label{theorem:edgePartitioning}
	\TheoremEdgePartitioning
\end{theorem}

The latter procedure is based on a \emph{BFS-like} approach and allows each node to double the hop-distance of the neighborhood which it knows in $O(1)$ rounds, \emph{at least} until the first cycle is encountered. This is stated formally in Theorem~\ref{theorem:bfsDoubling}

\newcommand{\TheoremBFSDoubling}
{
	Let $G$ be a graph with $n$ nodes, an unknown girth $g < \log n$, and with minimum degree $\delta \geq 2$. Assume that for a given integer parameter $a > 0$, every $v \in V$ knows the edges of $N_{a}(v)$, and that $N_a(v)$ is a tree. There exists an algorithm which completes in $O(1)$ rounds of the \clique model, and either reports $g$ exactly, or, reaches one of the following two states: 
	(1) Every $v \in V$ knows $N_{2a}(v)$, \emph{or}
	(2) For some value $b\geq \lceil \frac{g}{2} \rceil-1$ which is agreed upon by all nodes of the network, every $v \in V$ knows all of $N_b(v)$.
	All nodes know whether $g$ was reported exactly, and if not, which state was reached.
}
\begin{theorem}
	\label{theorem:bfsDoubling}
	\TheoremBFSDoubling
\end{theorem}

Thus, by invoking the first algorithm once, and then the latter for a constant number of times, we achieve an $O(1)$ round algorithm for the approximation problem. The reason we achieve a $+1$ approximation, and not an exact result, is due to the fact that the second algorithm stops right before detecting the shortest cycle in the graph and cannot differentiate whether it is of odd or even length. 
In Section~\ref{section:thrmGirthApproxProof}, we formally prove Theorem~\ref{theorem:girthApprox}, when $g < \log n$ and the minimum degree is $\delta \geq 2$, using Theorems~\ref{theorem:edgePartitioning},~\ref{theorem:bfsDoubling}. The constraints $g < \log n$ and $\delta \geq 2$ can be quickly overcome by eliminating some trivial, degenerate cases, and it is shown how to remove these constrains in Section~\ref{section:preliminaryPreprocessing}. Finally, in Sections~\ref{section:phase1},~\ref{section:phase2}, we proceed to our fundamental technical contributions by showing the proofs of Theorems~\ref{theorem:edgePartitioning},~\ref{theorem:bfsDoubling}.

\subsection{Preliminary Preprocessing}
\label{section:preliminaryPreprocessing}
Prior to the initiation of the main algorithm which achieves the $+1$ approximation for girth in the \clique model, we perform some preliminary steps in order to treat trivial or degenerate cases. Specifically, we take care of the case when $g \geq \log n$, and ensure that the minimum degree in the graph is $\delta \geq 2$, without changing the girth of the graph.

We assume that the graph is simple, i.e., does not contain self-loops or multiple edges. Notice that such cases are trivial.

\subparagraph*{Graphs With High Girth}
We note that by~\cite[Theorem 4.1]{ZM13}, if $g \geq \log{n}$ or $G$ has no cycles, then $m = O(n)$. Thus, in this case, by using the routing algorithm of~\cite{Lenzen13}, all the nodes in the graph can learn the entire graph in $O(1)$ rounds and output $g$ using local computation. Therefore, for the remainder of the algorithm, we may assume that $g < \log{n}$.

\subparagraph*{Degenerate Nodes}
We remove all nodes that do not participate in any cycles. In particular, after the removal of these nodes, no node $v$ with $d(v) < 2$ remains. This procedure can be seen a specific case of procedures used in \cite{KKM13,AGM12,GP16}.

\begin{definition}[$1$-degenerate nodes]
	A node is called $1$-degenerate if it is marked in the following process. Mark all nodes with $d(v) < 2$; remove all marked nodes; repeat as long as it is possible to mark nodes. 
\end{definition}

\begin{lemma}
	A $1$-degenerate node does not participate in any cycle.
\end{lemma}
\begin{proof}
	For a cycle $C$, assume by contradiction there is such a node. Let $v$ be the first node in $C$ removed by the process and let $t$ be the time at which it is removed. Since no other node was removed prior to time $t$, then both neighbors of $v$ in $C$ are still part of the graph at time $t$. Therefore, $d(v) \geq 2$ at time $t$, contradicting the claim that it was removed at that time. 
\end{proof}

The network can detect all $1$-degenerate vertices in $G$ in $O(1)$ rounds in the following manner. Each node $v$ broadcasts $(v.id,d(v),\bigoplus_{(u,v) \in E} u.id)$, which are overall $O(\log{n})$ bits. Each node locally and iteratively does the following process until there are no nodes of degree $1$: \emph{If there is a node $v$ of degree $1$ in the graph, $\bigoplus_{(u,v) \in E} u.id$ is just the ID of its only neighbor. For that $u$, decrease the degree of $u$ by $1$ and XOR the third field of $u$ with $v.id$. Remove $v$ from the graph.}

\subsection{Proving Theorem~\ref{theorem:girthApprox}}
\label{section:thrmGirthApproxProof}

Here, we prove Theorem~\ref{theorem:girthApprox}, in case that $g < \log n$ and the minimum degree is $\delta \geq 2$.
\begin{proofof}{Theorem~\ref{theorem:girthApprox}}
	
	We first invoke the algorithm from Theorem~\ref{theorem:edgePartitioning}, in $O(1)$ rounds. Either $g$ was outputted, or, every node learned its $\lfloor g/20 \rfloor$ hop-neighborhood.
	
	Next, we invoke the algorithm from Theorem~\ref{theorem:bfsDoubling}. The nodes now learned new, larger neighborhoods - regardless of whether the algorithm halted in State 1 (every $v \in V$ knows $N_{2a}(v)$) or State 2 (for some $b\geq \lceil \frac{g}{2} \rceil-1$, every $v \in V$ knows all of $N_b(v)$). If any node sees a cycle, then it broadcasts the length of the shortest cycle which it sees and all the nodes terminate and output the minimum of the values which were broadcast in the network. It is clear that, in this case, the exact value of $g$ is outputted, since all nodes know the  neighborhoods surrounding them of same radius, and thus if any node saw a cycle, one node must have seen the shortest cycle in the graph.
	
	Finally, in the case that no cycle was seen so far, we differentiate between the states at which the algorithm from Theorem~\ref{theorem:bfsDoubling} can halt at. If it halts at State 1, then every node learned twice the radius of the neighborhood it already knew. In such a case, we invoke Theorem~\ref{theorem:bfsDoubling} again and repeat. Notice that we can do this at most $O(1)$ times, before either seeing a cycle or halting at State 2, due to the fact that the nodes originally know their $\lfloor g/20 \rfloor$ hop-neighborhoods. In the case that we eventually halt at State 2, and no cycles were seen by any node so far, all the nodes output that the girth is either $2b+1$ or $2b+2$, where $b$ is the radius of the neighborhoods which they learned. It is clear that if all nodes learned their $b\geq \lceil \frac{g}{2} \rceil-1$ hop-neighborhoods, and none saw a cycle, then it must be that $b = \lceil \frac{g}{2} \rceil-1$ and thus either $g=2b+1$ or $g=2b+2$.
\end{proofof}

\subsection{Phase I: Initial Neighborhood Learning}
\label{section:phase1}
The key procedure of this phase (formally stated above as Theorem~\ref{theorem:edgePartitioning})
consists of two major steps. In the first step, either each path of length $\lfloor \frac{g}{10} \rfloor$ in $G$ is detected by at least one node, or $g$ is outputted. This step can be seen as an edge-partition variant of the listing algorithm in~\cite{DLP12}. The second step uses a load-balancing routine in order to redistribute the information computed in the first step so that each node $v$ learn its $\lfloor \frac{g}{20} \rfloor$ hop-neighborhood.

\subparagraph*{Step 1: Path Listing.}
We next list all paths of length $\lfloor \frac{g}{10} \rfloor$, or output $g$, in $O(1)$ rounds.

First, each node sends its degree to the rest of the network, and each then locally calculates the number of edges in the graph $m = \sum_v deg(v)/2$. Let $k \in \mathbb{N}$ be the largest integer such that $m \leq n^{1+1/k} + n$. Then, each node $v$ is assigned a hard-coded range of $deg(v)$ indices in $\set{1, \dots, m}$, and locally numbers its edges using these indices.

If $k \leq 4$, then by Lemma~\ref{lem:prelim_girth_turan} the girth is of size at most $10$, and thus, trivially, paths of length $\lfloor \frac{g}{10} \rfloor \leq 1$ are known and we can halt. Thus, from here on, we may assume that $k\geq 5$.

Let $P$ be a partition of the set $\{1,\dots,m\}$ into $\lceil kn^{2/k}/(20e) \rceil$ consecutive segments of size $O\left(\frac{m}{\lceil kn^{2/k}/(20e) \rceil} + 1\right)$, and let $K$ be a family containing all the possible choices of $\lfloor k/4 \rfloor$ segments from $P$ (in this context, $e$ denotes the mathematical constant). It holds that 
$$|K| = {\lceil kn^{2/k}/(20e) \rceil \choose \lfloor k/4 \rfloor} \leq \left(\frac{e\lceil kn^{2/k}/(20e) \rceil}{\lfloor k/4 \rfloor}\right)^{\lfloor k/4 \rfloor} \leq \left(\frac{n^{2/k}}{2}+1\right)^{\lfloor k/4 \rfloor} \leq n^{\frac{2}{k}\lfloor k/4 \rfloor}  \leq n,$$
where the first inequality holds due to the well-known combinatorial statement that ${n \choose k} \leq (\frac{ne}{k})^k$, for all $n \in \mathbb{N}, 1 \leq k \leq n$, and in the other inequalities, the fact that $5 \leq k < \log n$ is used. Thus, it is possible to associate each $k_i \in K$ with a unique node $v_i$. Each $k_i$ is a set of $\lfloor k/4 \rfloor$ sets of $O\left(\frac{m}{\lceil kn^{2/k}/(20e) \rceil} + 1\right)$ edges, and so let $E_i$ denote the edges in the sets contained in $k_i$. Notice that 
$$|E_i| \leq \lfloor k/4 \rfloor\left(\frac{m}{\lceil kn^{2/k}/(20e) \rceil} + 1\right) \leq (k/4) \frac{20en^{1+1/k}}{kn^{2/k}} + k= 5en^{1-1/k} + k \leq n,$$ 
and therefore, by Corollary \ref{cor:load_balancing_routing_deterministic}, it is possible for each $v_i$ to learn all of the edges in $E_i$ in $O(1)$ rounds of communication.

Finally, every node $v_i$ broadcasts the shortest cycle which it witnesses in $E_i$. Notice that every path, $p$, of length at most $\lfloor k / 4 \rfloor$, is fully contained inside some $E_j$, due to the construction of $P$, and therefore the corresponding node, $v_j$, which now knows all of $E_j$, will witness $p$. Thus, if $g \leq \lfloor k / 4 \rfloor$, some node will witness the shortest cycle in the graph and be able to broadcast its length, $g$. Otherwise, notice that since $m\not\le n^{1+1/(k+1)}+n$, Lemma~\ref{lem:prelim_girth_turan} implies that the graph is not $C_i$-free for all $i\le 2(k+1)$, and thus $g\le 2(k+1)$. Thus all paths of length at most $\lfloor k / 4 \rfloor \geq \lfloor g/8 - 1/4 \rfloor \geq \lfloor g / 10\rfloor$ have been listed. Notice that $g/8 - 1/4 \geq g/10$ whenever $g \geq 10$, and this can be assumed, since otherwise, trivially, paths of length $\lfloor \frac{g}{10} \rfloor < 1$ are known. 

If at least one node $v_i$ informs about a cycle in $E_i$, the minimum number sent by a node is outputted as $g$, and the algorithm terminates. Otherwise, it proceeds to the second step.

\subparagraph*{Step 2: Neighborhood Learning.}
We desire to redistribute some of the information learned in the previous step so that each node will know its $\lfloor \frac{g}{20} \rfloor$ hop-neighborhood.

Notice that all paths of length at most $\lfloor g / 10\rfloor$ have been listed. Therefore, also all paths of length at most $\lfloor \frac{g}{20} \rfloor$ have been listed. We strive to redistribute this information so that each node $v$ knows all paths of length at most $\lfloor \frac{g}{20} \rfloor$ which start at $v$, and thus $v$ knows its entire $\lfloor \frac{g}{20} \rfloor$ hop-neighborhood. Notice that we would like for each $v$ to know both the nodes in its $\lfloor \frac{g}{20} \rfloor$ hop-neighborhood, as well as the edges between them.

We begin by ensuring that each $v$ knows every node $u$ in its $\lfloor \frac{g}{20} \rfloor$ hop-neighborhood. Let $v \in V$ and $u$ be some node in its $\lfloor \frac{g}{20} \rfloor$ hop-neighborhood. Since $\lfloor \frac{g}{20} \rfloor < g / 2$, then the $\lfloor \frac{g}{20} \rfloor$ hop-neighborhood of $v$ is a tree. Therefore, there exists exactly one path, $p_{v, u}$, of length at most $\lfloor \frac{g}{20} \rfloor$ between $v$ and $u$. In the previous step, we ensured that at least one node $w$ is aware of $p_{v, u}$. Specifically, notice that it might be the case that many nodes know about $p_{v, u}$, due to the last step, yet, every node $w$ which knows of this path also knows all the other nodes $w'$ which learned this path through their $E_{w'}$. Thus, it is possible to choose, in a hard-coded manner, a single node $w$ which will be responsible for informing $v$ that $p_{v, u}$ exists. Having done that, node $w$ desires to convey to $v$ the message that $u$ is in its $\lfloor \frac{g}{20} \rfloor$ hop-neighborhood, in addition to the hop-distance between $v$ and $u$ --- that is, the length of $p_{v, u}$. Notice that for each such $u$, node $v$ is destined to receive exactly one message, and therefore every node in the graph is the target of $O(n)$ messages. This shows that Corollary~\ref{cor:load_balancing_routing_deterministic} may be invoked in order to deliver all these messages in $O(1)$ rounds. 

Now, we desire to inform every $v$ of the edges in its $\lfloor \frac{g}{20} \rfloor$ hop-neighborhood. Node $v$ now knows all the nodes $u$ in this neighborhood, as well as the hop-distance to each of them. Node $v$ sends a message to each such $u$ which is at most $\lfloor \frac{g}{20} \rfloor - 1$ hops away from it, and requests that $u$ send to $v$ \emph{all} its incident edges in the graph. Notice that all these edges are exactly all the edges contained in the $\lfloor \frac{g}{20} \rfloor$ hop-neighborhood of $v$, and since this neighborhood is a tree, $v$ is the target of at most $O(n)$ messages. As before, this shows that Corollary~\ref{cor:load_balancing_routing_deterministic} may be invoked in order to deliver all these messages in $O(1)$ rounds.

\subsection{Phase II: Neighborhood Doubling}
\label{section:phase2}
The key procedure in this phase (formally stated above as Theorem~\ref{theorem:bfsDoubling}) is an $O(1)$ round algorithm which doubles the radius of the hop-neighborhood known to each node, until the nodes know a neighborhood large enough in order to approximate the girth up to an additive value of 1. 
%
%
The algorithm works along the following lines. Denote by $F_a(v)$, the nodes which are exactly at distance $a$ from $v$ --- we refer to these as the \emph{front-line} nodes. Each nodes $v$ initially knows $N_a(v)$, and at once attempt to learn all of $\bigcup_{u \in F_a(v)} (N_a(u) \setminus N_a(v))$, in an efficient manner. If this step succeeds, then all the nodes reach State 1, and halt. Otherwise, they coordinate to increase the radii of the neighborhoods which they know by as much as possible in $O(1)$ rounds, and ultimately arrive at State 2, and halt.

\subparagraph*{Halting at State 1.}
Let $v \in V$ and $u \in F_a(v)$. Node $u$ aims to send to node $v$ the edges in $N_a(u) \setminus N_a(v)$. Notice that node $u$ can locally compute these edges as follows. It observes the first node $w$ on the path between $v, u$. Since $N_a(u)$ is a tree, for every node $w' \in N_a(u)$, there is exactly one simple path, $p_{u, w'}$, which $u$ sees to $w'$. Notice that $w' \in N_a(v)$ if and only if $p_{u, w'}$ passes through $w$. Thus, node $u$ knows exactly which edges it desires to send to node $v$. However, before sending them, it first sends to node $v$ the value $|N_a(u) \setminus N_a(v)|$.

We now shift back to the perspective of node $v$. It  computes and broadcasts an upper bound on $|\bigcup_{u \in F_a(v)} (N_a(u) \setminus N_a(v))|$, by calculating $\sum_{u \in F_a(v)} |N_a(u) \setminus N_a(v)|$. If all nodes broadcast values which are at most $n - 1$, then by Corollary~\ref{cor:load_balancing_routing_deterministic}, it is possible in $O(1)$ rounds to perform all the routing requests and have each node double the radius of the neighborhood which it knows. At this point, the nodes collectively reach State~1 and halt.

Otherwise, at least one node reported a value greater than or equal to $n$. This implies that for some node $v$, there is a cycle in $N_{2a}(v)$, since at least two nodes $u, u' \in F_a(v)$ have simple paths in their $a$ hop-neighborhoods to the same node $w$. In this case, the nodes proceed to a second part of the algorithm, which eventually leads to halting at State 2.

\subparagraph*{Halting at State 2.}
Our goal, at this stage, is to determine the largest possible value $i' \in \{1, \dots, a - 1\}$, such that for every node $v$, $\sum_{u \in F_a(v)} |N_{i'}(u) \setminus N_a(v)| < n$. Once this is achieved, then the algorithm can complete in a similar manner to that above. To see this, assume that we have this maximal value $i'$. Therefore, all nodes $v$ can learn $N_{a + i'}(v)$ in $O(1)$ rounds, similarly to above. If any cycle is seen, then $g$ is outputted and the algorithm halts. Otherwise, due to the definition of $i'$, there must exist some node $v'$ such that $\sum_{u \in F_a(v')} |N_{i'+1}(u) \setminus N_a(v')| \geq n$. This implies that there is a cycle in $N_{a + i' + 1}(v')$, and therefore $2a + 2i' < g \leq 2a + 2i' + 2$. As such, $a + i' = (2a' + 2i' + 2)/2 - 1 \geq \lceil g/2 \rceil - 1$, and we may halt at State 2.

We now show how to find $i'$. This is trivially possible to accomplish in $O(a)$ rounds --- each node $u$ simply sends to $v$ the values $\{ |N_1(u) \setminus N_a(v)|, \dots, |N_{a-1}(u) \setminus N_a(v)| \}$, node $v$ locally computes the $a-1$ different sums, and broadcasts them. However, this does not suffice for our goal of an $O(1)$ round algorithm, as $a$ can be logarithmic in $n$. Instead, let every node $v$ broadcast $|F_a(v)|$, and denote by $v'$ the node with maximal $|F_a(v')|$, and write $d = \left\lfloor n/|F_a(v')| \right\rfloor$. For every node $v$ and $u \in F_a(v)$, node $u$ sends to $v$ the values  $\{ |N_1(u) \setminus N_a(v)|, \dots, |N_{d}(u) \setminus N_a(v)| \}$, node $v$  computes the $d$ different sums of these values from all $u \in F_a(v)$, and broadcast them. Notice that this takes $O(1)$ rounds, using Corollary~\ref{cor:load_balancing_routing_deterministic} as each node wants to receive at most $O(n)$ messages. Notice that it is now possible in $O(1)$ rounds to compute $\min\{i', d\}$ --- either $i' \leq d$, and thus $ \min\{i', d\} = i'$ and we can compute it, or, $ \min\{i', d\} = d$. If we show that $g \leq 2a + 2d$, then if all $v$ learn $N_{a + \min\{i', d\}}(v)$, this would suffice in order to either find the exact girth or halt at State 2, as required.

We claim that $g \leq 2a + 2d$. To see this, assume that $g > 2a + 2d$. Since $g > 2a + 2d$, there are no cycles in $N_{a+d}(v')$. Combining this with the fact that we assume the minimal degree in $G$ to be at least 2, we can see that for all $j\neq j' \in \set{1, \dots, d}$, it holds that $|F_{a+j}(v')| \geq |F_{a+j-1}(v')|$, and $F_{a+j}(v') \cap F_{a+j'}(v') = \emptyset$. Thus, in $N_{a+d}(v')$ there are at least $\left(d + 1\right) \cdot |F_a(v')| = \left(\left\lfloor n/|F_a(v')| \right\rfloor + 1\right) \cdot |F_a(v')| > n$ nodes, a clear impossibility. As we have arrived at a contradiction, we get that $g \leq 2a + 2d$, as required. 

\section{Subgraph listing in the \clique model}
\label{section:listing}

We show an efficient ``sparsity-aware'' algorithm to \emph{list} subgraphs in the \clique model. Our main result is the following theorem, which is proven in the following sections.

\begin{theorem-repeat}{theorem:generalSubgraphTheorem}
	\TheoremSubgraphsCC
\end{theorem-repeat}

As mentioned in the introduction, we can combine this result with known bounds on the number of edges in graphs without specific subgraphs, to achieve fast subgraph \emph{detection} results. First, by combining Theorem~\ref{theorem:generalSubgraphTheorem} with Lemma~\ref{lem:prelim_turan}, we immediately get Corollary~\ref{corollary:CyclesCC}: If the graph contains more than $17kn^{1+1/k}$ edges (which can be checked in a single round), then by Lemma~\ref{lem:prelim_turan} we can safely output that there must exist a cycle of length $2k$. Otherwise, plugging $p=2k$ and $m=17kn^{1+1/k}$ in Theorem~\ref{theorem:generalSubgraphTheorem} gives that we can detect (and even list, in this case) the existence of a cycle of length $2k$ within $O(k^2)$ rounds. 
Next, by combining Theorem~\ref{theorem:generalSubgraphTheorem} with Lemma~\ref{lem:prelim_girth_turan}, we can get the following result:

\newcommand{\ExactGirth}
{
	Given a graph $G$ with $n$ nodes, $m$ edges and an unknown girth $g$ such that $g > \ell$ for some known $\ell$, and defining $f(x) = n^{1/\lfloor(x-1)/2\rfloor - 2/(2 \cdot \lfloor(x-1)/2\rfloor + 1)}$, there is a deterministic $\tilde{O}(\min\{f(g),f(2 \cdot \lfloor(\ell + 1) / 2\rfloor + 1)\})$ round algorithm in the \clique model which outputs g.
}

\begin{proposition}
	\label{proposition:exactGirth}
	\ExactGirth
\end{proposition}
Proposition \ref{proposition:exactGirth} first shows that the exact girth can be computed in $\tilde{O}(f(g))$ rounds --- a polynomial improvement over the state-of-the-art for all graphs with $g \geq 5$. Moreover, if it is known that the graph has girth greater than $\ell$,\footnote{It is possible to phrase a slightly stronger result which does not require a lower bound on the girth, but rather that for a specific $k$, which depends on the sparsity of the graph, there will not be any cycles of length $2k+1$.} then the round complexity is additionally guaranteed to be $\tilde{O}(f(2 \cdot \lfloor(\ell + 1) / 2\rfloor + 1))$. For instance, for any odd value $\ell=2r-1$ we get the upper bound $\tilde O(n^{1/r-2/(2r+1)})$. Taking $r=2$ gives Corollary \ref{corollary:triangle-free-graphs} stated in the introduction, which improves upon the state-of-the-art for triangle free graphs. 

We note that more claims can be shown using bounds for the T{\'u}ran numbers of various other graphs --- for example, for detection of $K_{s, t}$ (complete bipartite graph with $s$ nodes on one side and $t$ on the other) for certain values of $s, t$.

\begin{proofof}{Proposition~\ref{proposition:exactGirth}}
	Let $k'$ be the largest integer such that $m \leq n^{1 + 1/k'} + n$. If $k' \geq (\log n)/2$, then $m = O(n)$ and thus the entire graph can be learned by one node in $O(1)$ rounds, completing the proof.
	
	Otherwise, it is known that a cycle of length at most $2k' + 2$ exists in the graph, due to Lemma~\ref{lem:prelim_girth_turan} and $m > n^{1 + 1/(k'+1)} + n$ due to the definition of $k'$.
	Notice that since $m \leq n^{1 + 1/k'} + n$, then for each $p \leq 2k'$, it is possible to list all $C_p$ in the graph in $O(k')$ rounds using Theorem~\ref{theorem:generalSubgraphTheorem}. Therefore, since $k' < (\log n)/2$, it is possible in $\tilde{O}(1)$ rounds to list all cycles of length up to $2k'$. If a cycle is witnessed at this stage, then the nodes know the exact girth of the graph and halt.
	
	We arrive at the last case, which is determining whether a cycle of length $2k' + 1$ exists. We invoke Theorem~\ref{theorem:generalSubgraphTheorem} to list all $C_{2k' + 1}$, which takes $\tilde{O}(n^{1/k' - 2/(2k'+1)})$ rounds, and allows the nodes to determine the exact girth of the graph. Notice that the girth is either $2k'+1$ or $2k'+2$, and thus $f(g) = f(2k'+1) = f(2k'+2) = 1/k' - 2/(2k' + 1)$. The overall round complexity of the algorithm is thus $\tilde O(f(g))$.
	
	We now consider the case where we additionally know that $g>\ell$, and derive another bound on the complexity that depends only on $\ell$. If $\ell$ is even then we simply run the above algorithm; the complexity is $\tilde O(f(2 \cdot \lfloor(\ell + 1) / 2\rfloor + 1))$ since $g\ge \ell+1=2 \cdot \lfloor(\ell + 1) / 2\rfloor + 1$. Now assume that $\ell$ is odd. In that case we first check if the graph is $C_{\ell+1}$-free in $\tilde O(1)$ rounds using the algorithm of Corollary \ref{corollary:CyclesCC}. If the graph is not $C_{\ell+1}$-free, then we know that $g=\ell+1$. Otherwise we know that $g\ge \ell+2$ and we run the above algorithm; the complexity is again $\tilde O(f(2 \cdot \lfloor(\ell + 1) / 2\rfloor + 1))$ since $g\ge \ell+2=2 \cdot \lfloor(\ell + 1) / 2\rfloor + 1$.
\end{proofof}

\subsection{Partition trees}
\label{subsection:partitionTrees}

We introduce the notion of \emph{partition trees}, as a fundamental tool for subgraph listing in the \emph{deterministic} \clique model. Partition trees are a deterministic load-balancing mechanism that evenly divides the work of checking whether any copies of a subgraph are present. In prior work, randomized load-balancing was used for this purpose, but this incurs logarithmic factors which we cannot tolerate here.
Throughout this section, given a subgraph with $p$ nodes, we frequently refer to the value $x = n^{1/p}$. We assume that $x$ is an integer, because $p \leq \log n$ implies $x \geq 2$, and so it is possible to round $x$ to an integer without affecting the round complexity or correctness.

We start with Definition~\ref{def:partitionTreeP}, which defines a $p$-partition tree, which is a tree structure in which every node represents a partition of the graph $G$. Then, in Definition~\ref{def:partitionTreeH}, we define an $H$-partition tree, in which we require certain conditions on the number of edges between parts of a $p$-partition, based on the subgraph $H$ of $p$ nodes which we will want to list.
\begin{definition}[$p$-partition tree, Figure~\ref{fig:partitionTree}]
	\label{def:partitionTreeP}
	Let $G=(V,E)$ be a graph with $n$ nodes and $m$ edges, and let 
	$p \leq \log n$.  
	A \emph{$p$-partition tree} $T=T_{G,p}$ is a tree of $p$ layers (depth $p-1$), where each non-leaf node has at most $x=n^{1/p}$ children. Each node in the tree is associated with a partition of $V$ consisting of at most $x$ parts.
	
	We inductively denote all partitions associated with nodes in $T$ as follows. The partition associated with the root $r$ of $T$ is called the \emph{root partition}, and is denoted by $P_{\emptyset}$. Given a node with a partition denoted by $P_{(\ell_1,\dots,\ell_{i-1})}$, the partition associated with its $j$th child, for $0\leq j\leq x-1$, is denoted $P_{(\ell_1,\dots,\ell_{i-1},j)}$. 
	
	\sloppypar{The at most $x$ parts of each partition $P_{(\ell_1,\dots,\ell_{i})}$ are denoted by $U_{(\ell_1,\dots,\ell_{i}), j}$, for $0\leq j\leq x-1$. For each $0 \leq j \leq x-1$, the part $U_{(\ell_1,\dots,\ell_{i-1}), \ell_i}$ is called the \emph{parent} of the part $U_{(\ell_1,\dots,\ell_{i-1},\ell_{i}), j}$, also denoted as $U_{(\ell_1,\dots,\ell_{i-1}), \ell_i}=\parent(U_{(\ell_1,\dots,\ell_{i-1},\ell_{i}), j})$.}
\end{definition}
\begin{figure}[h]
	\begin{center}
		\includegraphics[trim = 2cm 9.5cm 8cm 2.5cm, clip, scale=0.4]{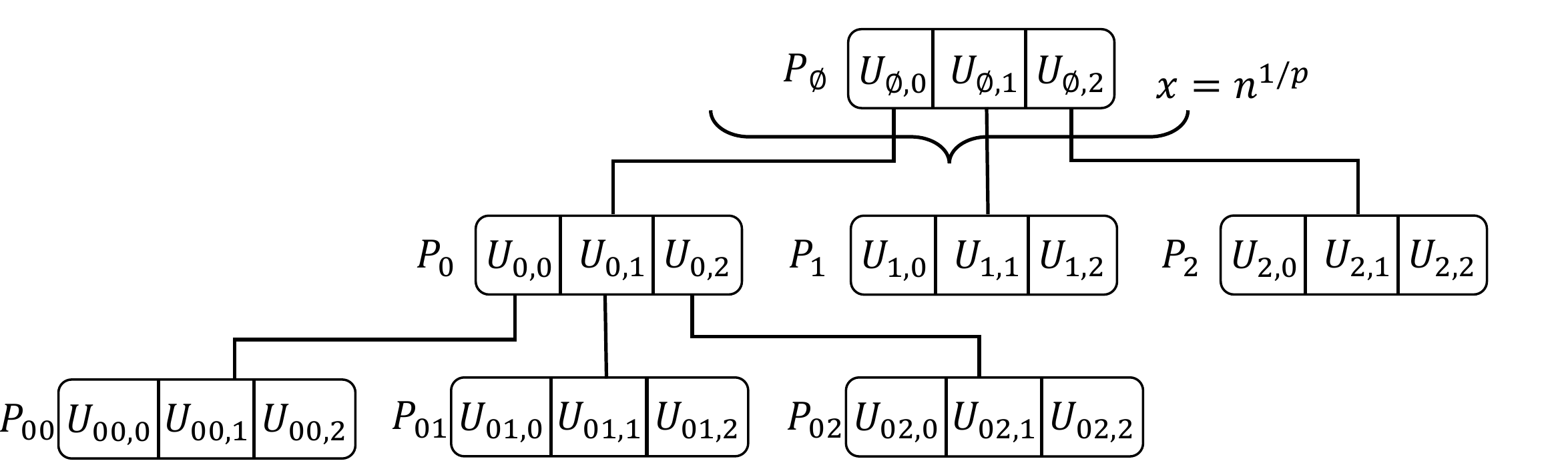}\vspace{-2mm} 
	\end{center}
	\caption{\label{fig:partitionTree} A partial illustration of a partition tree with $p,x=3$.}
\end{figure}

\begin{definition}[$H$-partition tree]
	\label{def:partitionTreeH}
	Let $G=(V,E)$ be a graph with $n$ nodes and $m$ edges, and let $H$ be a graph with $p \leq \log n$ nodes, $\{z_0,\dots, z_{p-1}\}$, and denote $d_i = |\{\{z_i,z_t\} \in E_H\mid t < i\}|$ for each $0\leq i \leq p-1$, $x=n^{1/p}$ and $\tilde{m}=\max\{m,nx\}$. A \emph{$H$-partition tree} $T=T_{G,H}$ is a $p$-partition tree with the following additional constraints, for some constants $c_1,c_2$. 
.
	\begin{enumerate}
		\item\label{EdgesUV} for every part $U=U_{(\ell_1,\dots,\ell_{i-1},\ell_{i}), j}$, it holds that  $|E(U, V)| \leq c_1 m/x +n$, and 
		\item\label{EdgesUParent} for every part $U_i=U_{(\ell_1,\dots,\ell_{i-1},\ell_{i}), j}$, and all of its ancestor parts $U_t = \parent(U_{t+1})$ for $t = i-1,\dots 0$, it holds that  $\sum_{t<i, \{z_i,z_t\} \in E_H}{|E(U, U_t)|} \leq c_2 d_i \tilde{m}/x^2 +n$,
	\end{enumerate}
\end{definition}

Notice that in Definition~\ref{def:partitionTreeH}, we define $\tilde{m}$ as an upper bound on $m$, the number of edges in the input graph. This is done as if the graph is \emph{too} sparse, we use a slightly higher bound on the number of edges in order to make decisions regarding the constraints on the partitions. We note that $\tilde{m}$ is purely a technicality --- we do not \emph{require} that there be at least this many edges in the graph.
In the following two theorems: we show that we can construct an $H$-partition tree and use it to efficiently perform $H$-listing. 

\newcommand{\TheoremPartitionTree}
{
	Let $G=(V,E)$ be a graph with $n$ nodes, and let $H$ be a graph with $p \leq \log n$ nodes. There exists a deterministic \clique algorithm that completes in $O(1)$ rounds and constructs an $H$-partition tree $T$, such that $T$ is known to all nodes of $G$ --- that is, all nodes know all the partitions making up $T$.
}
\begin{theorem}
	\label{theorem:partitionTree}
	\TheoremPartitionTree
\end{theorem}
And second, that given an $H$-partition tree, we can list all instances of $H$ in $G$.
\newcommand{\TheoremListWithTree}
{
	Let $G=(V,E)$ be a graph with $n$ nodes, let $H$ be a graph with $p \leq \log n$ nodes and $k$ edges, and denote $x=n^{1/p}$ and $\tilde{m}=\max\{m,nx\}$. 	
	There exists a deterministic \clique algorithm that completes in $O(\frac{k\tilde{m}}{n^{1+2/p}} + p)$ rounds and lists all instances of $H$ in $G$, given an $H$-partition tree $T$ that is known to all nodes.
}
\begin{theorem}
	\label{theorem:listWithTree}
	\TheoremListWithTree
\end{theorem}

Thus, Theorems~\ref{theorem:partitionTree} and~\ref{theorem:listWithTree}, directly imply Theorem~\ref{theorem:generalSubgraphTheorem}.

\begin{proofof}{Theorem~\ref{theorem:partitionTree}}
	In order to show this proof, we construct a set of preliminary partitions in $O(1)$ rounds, and maintain that it is possible to construct the entire partition tree using only this set of partitions. By ensuring that these partitions are globally known, each node can compute the entire tree locally.

	\subparagraph*{Constructing a preliminary set of partitions.}
	We construct a main partition, $R$, with at most $x/2$ parts, and then several more partitions, \emph{of the entire graph}, which are refinements of $R$. Specifically, for every set of $1 \leq \ell \leq p-1$ parts, denoted $\{Q_{j_0}, \dots, Q_{j_{\ell - 1}}\}$, from $R$, we create a specific partition denoted as $M_{\{j_0, \dots, j_{\ell - 1}\}}$, which has at most $x$ parts $N_{\{j_0, \dots, j_{\ell - 1}\}, k}$ for $0 \leq k \leq x-1$. Notice that this is a total of at most $(x/2 + 1)^{p-1} \leq x^{p-1} = n/x$ different partitions. We emphasize that each $M_{\{j_0, \dots, j_{\ell - 1}\}}$ is a \emph{partition of the entire graph}, and not of $\{Q_{j_0}, \dots, Q_{j_{\ell - 1}}\}$.
	
	For each partition, we consider a set of $x$ nodes that are called \emph{the builder nodes}. We assign some $x$ nodes to build the main partition, denoted by $B_\emptyset$, and then we assign sets of builder nodes to each additional partition in a mutually disjoint manner. That is, denoting by $B_{\{j_0, \dots, j_{\ell - 1}\}}$ the set of builder nodes for a partition $M_{\{j_0, \dots, j_{\ell - 1}\}}$, gives that $B_{\{j_0, \dots, j_{\ell - 1}\}}\cap B_{\{j'_0, \dots, j'_{\ell - 1}\}} = \emptyset$ for every two such additional partitions. Due to the fact that there are at most $n/x$ additional partitions, it is clear that this assignment is possible. The builder nodes $B_\emptyset$ initially construct the main partition in $O(1)$ rounds, and then the additional partitions are constructed concurrently by their respective builder nodes in $O(1)$ rounds.
	
	For the main partition $R$, the only condition that we maintain is Condition~\ref{EdgesUV}, which requires that each of its  parts $Q$ satisfies $|E(Q, V)| \leq c_1 m/x +n$. To ensure this, each node $v$ sends its degree to all builder nodes in $B_{\emptyset}$. Then, the builder nodes go over the nodes in an arbitrary order (known to all nodes) and add them to parts of the (initially empty) partition, as follows. The first processed node $v$ is added to a part $Q_{0}$, and a counter is set to $deg(v)$. Then, every following node $v$ is added to $Q_{0}$ and the counter is increased by $deg(v)$, as long as it does not exceed $c_1m/x + n$. Once adding $v$ to a part would make the counter exceed the threshold, the next part  $Q_{1}$ is started, initialized to contain $v$ and its counter is $deg(v)$. We continue in this manner until all nodes are processed. 
	
	Notice that this creates at most $x/2$ parts in the partition by choosing $c_1 \geq 4$, since each part has at least $c_1 m/x$ edges out of $2m$ (counting each edge twice for both of its endpoints). Finally, note that the builder nodes in $B_{\emptyset}$ can inform all other nodes about the partition $R$ within $O(1)$ rounds, since there are at most $x/2$ parts that can each be described by their first and last nodes in the globally known order, and the description of each part can be broadcast to all nodes by a different builder node in $B_{\emptyset}$.
	
	When constructing the partition $M_{\{j_0, \dots, j_{\ell - 1}\}}$, we maintain three conditions. Primarily, we maintain Condition~\ref{EdgesUV}; that is, for every part $N=N_{\{j_0, \dots, j_{\ell - 1}\}, k}$ it holds that  $|E(N, V)| \leq c_1 m/x +n$. Furthermore, similarly to Condition~\ref{EdgesUParent}, we ensure that for each part, $N=N_{\{j_0, \dots, j_{\ell - 1}\}, k}$, $\sum_{0 \leq i < \ell}{|E(N, Q_{j_i})|} \leq c_2 \ell \tilde{m}/x^2 + n$. Lastly, we ensure that $M_{\{j_0, \dots, j_{\ell - 1}\}}$ is a refinement of $R$. 
	
	Each $M = M_{\{j_0, \dots, j_{\ell - 1}\}}$ is constructed in a similar manner to the way in which $R$ was constructed - every node $v$ in the graph sends some $O(1)$ messages to each builder node in $B = B_{\{j_0, \dots, j_{\ell - 1}\}}$, the builder nodes locally compute the partition $M$, and then each builder node broadcasts to the entire graph some part of $M$ in $O(1)$ rounds. To begin construction of $M$, every node $v$ sends the values $deg(v)$, $\sum_{0\leq t< \ell}{deg_{j_{t}}(v)}$ to all nodes in $B$, where $deg_{j_t}(v)$ is the number of neighbors of $v$ in $Q_{j_t}$. Similarly to the construction of the root partition, the builder nodes in $B$ go over the nodes in a known order and add them one by one to parts of the (initially empty) partition. In order to promise Condition~\ref{EdgesUV}, that $|E(N, V)| \leq c_1 m/x +n$ for every part $N$ that is constructed, a counter is maintained that accumulates the degrees $deg(v)$ of every node $v$ that is added to the current part. If adding a node $v$ would make this counter exceed the threshold, then a new part is started. To promise $\sum_{0 \leq i < \ell}{|E(N, Q_{j_i})|} \leq c_2 \ell \tilde{m}/x^2 + n$, for each part that is being constructed, a second counter is maintained. This counter accumulates $\sum_{0\leq t< \ell}{deg_{j_{t}}(v)}$ , for every $v$ that is added to the current part. If adding a node $v$ would make this counter exceed the threshold $c_2 \ell \tilde{m}/x^2 +n$, then a new part is started. Once a new part is started because adding a node $v$ would make one of the counters of the previous part exceed its threshold, the new part is initialized to contain $v$, and its counters are initialized to $deg(v)$ and $\sum_{0\leq t< \ell}{deg_{j_{t}}(v)}$ , respectively. We continue in this manner until all nodes are processed. Finally, in order to ensure $M$ is a refinement of $R$, we split every part in $M$ to parts completely contained in parts of $R$.

	We claim that this creates at most $x$ parts in each $M$ by choosing $c_1 = 8$ and $c_2 = 32$. Starting a new part can only happen due to one of the two counters exceeding its threshold, or due to a split of a part in order to ensure $M$ is a refinement of $R$. We first bound the number of parts created only according to the counters, and then proceed to the parts added due to splitting the parts of $M$ according to the parts of $R$. For the first counter, as in the analysis of the root partition, exceeding the threshold means that the part already has at least $c_1 m/x$ edges that touch it out of $2m$ possible edges counted for both endpoints. Therefore the first counter can exceed the threshold no more than $2x/c_1$ times. For the second counter to exceed its threshold, we have that the part already contains $c_2 \ell \tilde{m}/x^2$ edges to the relevant parts in $R$. Each of the corresponding parts in $R$ satisfies Condition~\ref{EdgesUV}  --- has at most $c_1 m/x+n$ edges touching it altogether --- and so in total there are at most  $c_1 \ell m/x + \ell n$ edges touching the corresponding parts in $R$.
	We thus claim the second counter can exceed its threshold no more than $c_1 x /c_2$ times. To see why, note that the second counter can exceed its threshold at most a number of times which is $\frac{c_1 \ell m/x + \ell n }{c_2 \ell \tilde{m}/x^2} \leq \frac{2 c_1 \ell \tilde{m}/x}{c_2 \ell \tilde{m}/x^2} =  2 c_1 x /c_2$. The final condition, that $M$ is a refinement of $R$, can add to $M$ at most the number of parts in $R$ - that is, at most $2x/c_1$ additional parts. To see this, notice that since all the partitions are created by going over all the nodes in the graph in some predetermined order and creating a new part once some counter has exceeded its threshold, then each part in $R$ can only incur a single additional point in time at which the builder nodes have to start a new part in $M$. Therefore, when setting $c_1 = 8, c_2 = 32$,  in total each $M$ has at most $4x/c_1 + c_1 x / c_2 \leq 3x / 4 \leq x$ parts.
	
	Finally, as each $M_{\{j_0, \dots, j_{\ell - 1}\}}$ has at most $x$ parts, the builder nodes $B_{\{j_0, \dots, j_{\ell - 1}\}}$ can ensure that $M_{\{j_0, \dots, j_{\ell - 1}\}}$ is globally known in $O(1)$ rounds.
	
	\subparagraph*{Locally constructing the entire tree.}
	We now show that using the preliminary set of partitions, each node can locally construct the entire partition tree. 
	
	We set the root partition as $P_\emptyset = R$, and proceed to setting the remaining layers of the tree. We begin by setting the first layer below the root partition. Each partition in this layer needs to maintain Condition~\ref{EdgesUParent} with respect to at most one part in $P_\emptyset$. Assume we need to construct $P_{(j)}$ for some $0 \leq j \leq x - 1$. We need to ensure that all of the parts in $P_{(j)}$ have a bounded number of edges entering part $U_{\emptyset, j}$. Thus, $M_{\{j\}}$ certainly maintains all the required conditions and we can set $P_{(j)} = M_{\{j\}}$. Next, we attempt to build the $i^{th}$ layer below the root partition. In this layer, every partition created, $P_{(j_0, \dots, j_{i-1})}$, has to maintain Condition~\ref{EdgesUParent} of Theorem~\ref{theorem:partitionTree} with respect to some subset of the parts $\{ U_{(j_0, \dots, j_{k-1}), j_k} | 0 \leq k < i \}$. However, since every $M_{\{j_0, \dots, j_{\ell - 1}\}}$ is a refinement of $P_\emptyset$, then each part in $\{ U_{(j_0, \dots, j_{k-1}), j_k} | 0 \leq k < i \}$ can be replaced by some part in $P_\emptyset$ which contains it, and thus if $P_{(j_0, \dots, j_{i-1})}$ maintains the required conditions w.r.t. a specific set of at most $i$ parts of $P_\emptyset$, then it would also maintain them w.r.t. $\{ U_{(j_0, \dots, j_{k-1}), j_k} | 0 \leq k < i \}$. We have already computed partitions which maintain all the required conditions with respect to any set of at most $p-1$ parts in $P_{\emptyset}$, and thus there exists a partition which we already computed in our preliminary set of partitions which can be used as $P_{(j_0, \dots, j_{i-1})}$.
\end{proofof}

\begin{proofof}{Theorem~\ref{theorem:listWithTree}}
	Denote by  $\{z_0,\dots, z_{p-1}\}$ the nodes of $H$, and denote $d_i = |\{\{z_i,z_t\} \in E_H\mid t < i\}|$ for each $0\leq i \leq p-1$. 
	
	We assign each leaf of the $H$-partition tree $T$ to $x$ different nodes. Note that there are $x^{p-1}$ leaves, which is at most $n/x$ due to our choice of $x=n^{1/p}$. We abuse the notation and denote a node in $T$ with the same notation as we use for the partition that is associated with it. Each leaf $P_{(\ell_1,\dots,\ell_{p-1})}$ is thus assigned to $x$ different nodes, and each part $U_{(\ell_1,\dots,\ell_{p-1}),j}$ in each leaf partition is assigned to a different node. For each node $v \in V$, we denote by $U_{v,p-1}$ the part of the leaf partition that it is assigned to. Then, inductively, for every $i = p-2,\dots 0$, we denote $U_{v,i} = \parent(U_{v,i+1})$. Note that for all $v \in V$ we have that $U_{v,0}$ is a part in the root partition.
	
	We now let every node $v \in V$ learn all the edges in $\bigcup_{t<i \text{ s.t. } \{z_i,z_t\} \in E_H}{E(U_{v,i}, U_{v,t})}$ and list all the instances of $H$ that it sees. We need to prove that all instances of $H$ in $G$ are indeed listed by this approach, and that learning the required edges by all nodes can be done in $O(\frac{k\tilde{m}}{n^{1+2/p}} + p)$ rounds.
	
	We first show that indeed all instances of $H$ are listed. Let $H'$ be an instance of $H$ in $G$, with nodes $\{z'_0,\dots, z'_{p-1}\}$, such that $\{z'_i,z'_t\}$ is an edge in $H'$ if and only if $\{z_i,z_t\}$ is an edge in $H$. Let $U^0$ be the part of the root partition that contains $z'_0$. Denote by $j_0$, where $0\leq j_0 \leq x-1$, the index of $U^0$ in the root partition, and let $P^1 = P_{(j_0)}$. Let $U^1$ be the part of $P^1$ that contains $z'_1$, and denote by $j_1$ the index of $U^1$ in $P^1$. Continue inductively, for $i = 2,\dots,p-1$: Let $P^{i} = P_{(j_0,j_1\dots,j_{i-1})}$. Let $U^{i}$ be the part of $P^{i}$ that contains $z'_{i}$, and denote by $j_{i}$ the index of $U^{i}$ in $P^{i}$. We now have a sequence of parts $U^{p-1}, U^{p-2},\dots,U^{0}$, and notice that for every $i$, $0\leq i \leq p-2$, we have that $U^{i} = \parent(U^{i+1})$. This implies that for the node $v \in V$ that is assigned to part $j_{p-1}$ of the leaf partition $P^{p-1}$, it holds that $U_{v,i}=U^{i}$ for every $0\leq i \leq p-1$, which means that $H'$ is contain in $\bigcup_{t<i \text{ s.t. } \{z_i,z_t\} \in E_H}{E(U_{v,i}, U_{v,t})}$, and thus the node $v$ indeed lists the instance $H'$ of $H$ given by $\{z'_0,\dots,z'_{p-1}\}$, as needed.
	
	It remains to bound the round complexity of having each node $v \in V$ learn about all of the edges in $\bigcup_{t<i \text{ s.t. } \{z_i,z_t\} \in E_H}{E(U_{v,i}, U_{v,t})}$. Since the $H$-partition tree $T$ satisfies Condition~\ref{EdgesUV} of Theorem~\ref{theorem:partitionTree}, we have that the number of edges that each node needs to learn is bounded by 
	\begin{eqnarray*}
		\label{eqn:learn}
		\bigcup_{t<i \text{ s.t. } \{z_i,z_t\} \in E_H}{E(U_{v,i}, U_{v,t})} &\leq& \sum_i{\sum_{t<i, \{z_i,z_t\} \in E_H}{|E(U, U_t)|}}\\
		&\leq & \sum_{i}{c_2 d_i \tilde{m}/x^2 + n}\\	
		&\leq & c_2(\sum_{i}{d_i}) \tilde{m}/x^2 + pn\\	
		&\leq & O(\frac{k\tilde{m}}{n^{2/p}} + pn). 
	\end{eqnarray*}
	Thus, by Corollary~\ref{cor:load_balancing_routing_deterministic}, all information can be learned in $O(\frac{k\tilde{m}}{n^{1+2/p}} + p)$ rounds, and so the algorithm completes in $O(\frac{k\tilde{m}}{n^{1+2/p}} + p)$ rounds, as claimed.
\end{proofof}

\section{Detecting Even Cycles and Computing the Girth in \congest}
\label{sec:app_congest}
In this section we present our $\congest$ algorithms for finding small even cycles and for computing the girth.

\subsection{Algorithm for Detecting Small Even Cycles}
Throughout, we assume the convention that negative indices are taken to be modulo the cycle size, that is, 
if we are working with cycles of length $\ell$, then we denote $u_{-i} = u_{\ell-i}$.
Likewise, when nodes choose random colors, the colors are numbers in $[\ell] = \set{0,\ldots,\ell-1}$,
but for convenience, we sometimes write $-i$ for color $\ell - i$.

Fix $k \in \set{2,3,4}$. We show that we can find a copy of $C_{2k}$, if there is one, in $O(n^{1-1/k})$ rounds, improving on previous algorithms, which had running time $n^{1-1/\Theta(k^2)}$.

We say that a $2k$-cycle $u_0,\ldots,u_{2k-1}$ is \emph{light} if each cycle node $u_i$ has degree at most $n^{1/k}$.
Otherwise we say that the cycle is \emph{heavy}.

\subsubsection{Finding Light Cycles}

Light cycles are easily found as follows:
repeat, for $R = \Theta((2k)^{2k})$ iterations, the following steps.
\begin{enumerate}
	\item Each node $u \in V$ chooses a random color $c(u) \in [2k]$. 
	\item We start a \emph{color-BFS} to depth $k$, in the subgraph of nodes that have degree $\leq n^{1/k}$:
		each node $u$ that has color $c(u) = 0$ and $\deg(u) \leq n^{1/k}$
		sends out a BFS token carrying its ID to all its neighbors that have color 1.
		Next, nodes with color $b \in \set{-1,+1}$ and degree $\leq n^{1/k}$ forward all the BFS tokens they receive to their neighbors with color $2b$; this requires at most $n^{1/k}$ rounds.
		We proceed similarly: in the $i$-th step of the BFS, nodes with degree $\leq n^{1/k}$ and color $b \cdot i$, where $b \in \set{-1,+1}$,
		forward all the BFS tokens they receive to their neighbors that have color $b \cdot (i+1)$.
		This requires at most $n^{i/k}$ rounds.
		Eventually, nodes with color $b \cdot (k-1)$ for $b \in \set{-1,+1}$ and degree $\leq n^{1/k}$ forward
		their BFS tokens to nodes with color $-k = k$.
	\item If a node with color $k$ receives the same BFS token from a neighbor with color $k-1$ and a neighbor with color $k+1=-(k-1)$,
		then it rejects.
\end{enumerate}

\paragraph*{Correctness.}
First, note that the algorithm never rejects unless the graph contains a copy of $C_{2k}$: 
for each $i \in \set{1,\ldots,k-1}$,
the BFS initiated by a node $u$ with $c(u) = 0$ can only reach a node $v$
with $c(v) \in \set{-i,+i}$ if there is a path of length $i$ from $u$ to $v$, whose nodes are colored $0,1,\ldots,i$ (if $c(v) = i$) or $0,-1,\ldots,-i$ (if $c(v) = -i$).
Therefore, a node $v$ with color $k$ rejects only if there is some node $u$ that has two disjoint length-$k$ paths to $v$, or in other words, node $v$ participates in a $2k$-cycle.

Next, suppose that the graph contains a light $2k$-cycle, $u_0,\ldots,u_{2k-1}$.
In a given iteration, with probability $1/(2k)^{2k}$, each cycle node $u_i$ chooses $c({u_i}) = i$.
Since the cycle is light, the number of BFS tokens that reach nodes $u_i, u_{2k-i}$ where $i \in \set{1,\ldots,k-1}$ is at most $n^{i/k}$:
since each cycle node has degree at most $n^{1/k}$,
there are at most $n^{i/k}$ nodes with color 0 that have a path of length $i$ to node $u_i$ or $u_{2k-i}$,
and as we said above, a given BFS token $u$ can only reach a node $v$ with color $i$ or $-i$ if there is a path of length $i$ from $u$ to $v$ (with ascending or descending colors).
This means that no cycle node is forced to stop participating in the middle of the BFS because it has too many tokens to forward.
The BFS token of node $u_0$ is forwarded in the color-ascending direction by $u_1,\ldots,u_{k-1}$ and in the color-descending direction by $u_{-1},\ldots,u_{-(k-1)}$ until it reaches node $u_k$, which rejects.

Since a given iteration succeeds with probability $1/(2k)^{2k}$,\footnote{Actually, the success probability is $1/(2k)^{2k-1}$, because we do not care about cyclic shifts of the colors on the cycle; but the next step of the algorithm does depend on getting the correct shift, so for simplicity we stick with the same number of iterations here as well.}
after $R = \Theta((2k)^{2k})$ iterations, we succeed with probability $2/3$.

\subsubsection{Finding Heavy Cycles}
It remains to find cycles where at least one node has degree greater than $n^{1/k}$.
To find such cycles, we exploit the fact that if we choose a random node in the graph, we have noticeable probability ($1/n^{1-1/k}$) of hitting a \emph{neighbor} of the cycle.
We show that with the exception of a small number of ``bad'' neighbors, if we find a neighbor of the cycle, we can find the cycle itself.

We first describe a ``meta-algorithm'' $\mathcal{A}$ that cannot quite be implemented in \congest,
analyze it, and then give an implementation $\mathcal{A}'$ in \congest;
the implementation is such that there is a high-probability event $\mathcal{U}$, conditioned on which $\mathcal{A}$ and $\mathcal{A}'$ are in some sense equivalent.

The meta-algorithm $\mathcal{A}'$ proceeds as follows:
let $T_k : \set{1,\ldots,2k-1} \rightarrow \nat$ be a function.
We repeat the following steps for $R'=\Theta(n^{1-1/k})$ iterations:
\begin{enumerate}
	\item We choose one uniformly random node $s \in V$.
		\label{step:choose_s}
	\item We carry out $R$ color-coded BFSs starting from $s$,
		each time using fresh independently-chosen colors for all the nodes in the graph.
		The BFS proceeds to depth $2k$, and it is only allowed to cross an edge $(u,v)$ if $c(v) = (c(u) + 1) \bmod 2k$.
		If one of the color-BFS instances finds a $2k$-cycle, we reject.
	\item Each node $u$ chooses a random color $c(u) \in [2k]$.
	\item We start a color-BFS from each neighbor of $s$ that has color $0$, in parallel.
		In step $i = 1,\ldots,k-1$ of the BFS, nodes colored $i$ or $-i$ (resp.) check if they have received more than $T_k(|i|)$ BFS tokens;
		if they have at most $T_k(|i|)$ tokens, they forward all of them to all neighbors colored $i+1$ or $-(i+1)$ (resp.),
		and if they have more than $T_k(|i|)$ tokens, they send nothing.
	\item If some node colored $k$ receives the same BFS token from neighbors colored $k-1$ and $k+1$,
		then it rejects.
\end{enumerate}
If after $R'$ iterations no node has rejected, then all nodes accept.
Note that the running time of the meta-algorithm is $R \cdot R' = \Theta(n^{1-1/k})$ (treating $k$ as a constant).

\subsubsection{High Level Overview of the Analysis}
When we search for heavy cycles, we sample a uniformly random node $s$, check if it is part of a $2k$-cycle,
and if not, we start a color-coded BFS from each 0-colored neighbor of $s$.
There can be many such neighbors, potentially leading to congestion;
however, we show that if the cycle is colored correctly,
it suffices for each node with color $i \in [2k]$ to forward a constant number $T_k(i)$ of BFS tokens.

Our main concern is that the node $s$ that we sampled is ``bad'', in the sense that it has many short node-disjoint paths to some cycle node $u_i$.
If we sample such a ``bad neighbor'' of $u_0$, its 0-colored neighbors could initiate many BFS instances, which would then reach $u_i$ and cause congestion.
See Figure~\ref{fig:bad_neighbor} for an illustration.

\begin{figure}[h]
     \centering
     \hfill
     \begin{subfigure}[t]{0.4\textwidth}
         \centering
	 \includegraphics[width=0.4\textwidth]{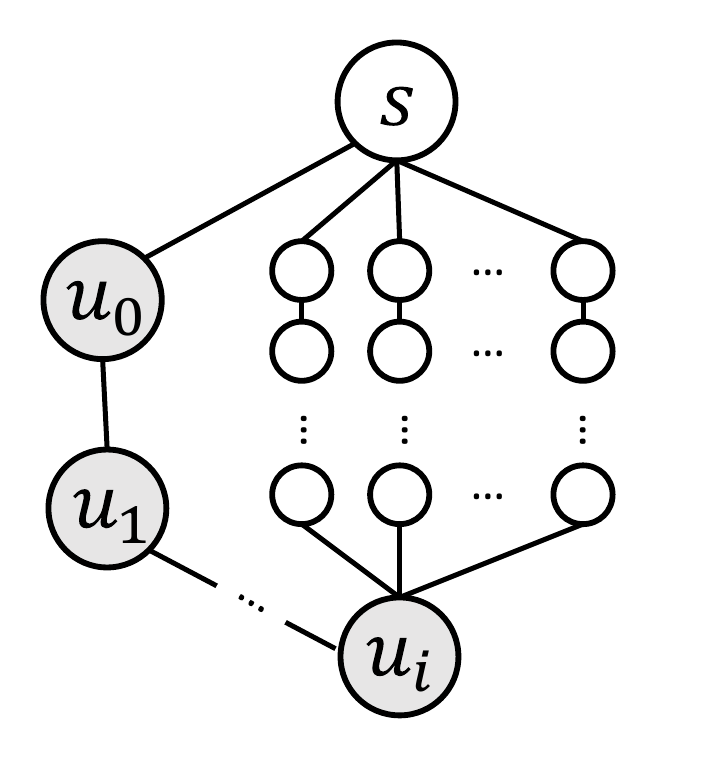}
	 \caption{A ``bad neighbor'' $s \in N(u_0)$.}
	 \label{fig:bad_neighbor}
     \end{subfigure}
     \hfill
     \begin{subfigure}[t]{0.45\textwidth}
         \centering
	 \includegraphics[width=0.4\textwidth]{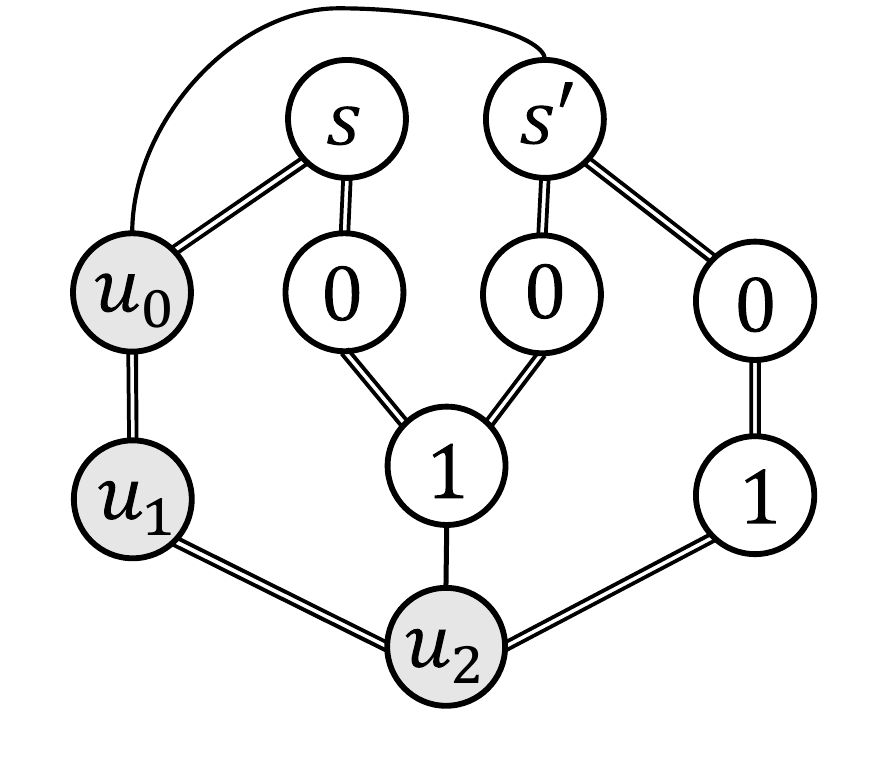}
	 \caption{``Shared paths'': the edges of the 10-cycle are indicated by double lines.}
	 \label{fig:bad_sharing}
     \end{subfigure}
     \hfill
     \caption{Illustrations for the proof sketch of the $2k$-cycle algorithm}
\end{figure}

To bound the probability that we hit a ``bad neighbor'', we first rule out any neighbor of $u_0$ that itself
participates in a $2k$-cycle.
Next, we argue that if $s$ has many node-disjoint paths to some cycle node $u_i$,
such that the path nodes are colored $0,1,\ldots,i$ (so that a BFS can be initiated by the first path node and flow across the path),
then we can charge these paths against the degree of $u_i$, as each path ends at a different neighbor of $u_i$.
Since $\deg(u_0) \geq \deg(u_i)$, this means only a small constant fraction of $u_0$'s neighbors have many such disjoint paths.
When we sample a random node, we are unlikely to hit a ``bad neighbor''.
(We are not worried about non-disjoint paths, as they do not contribute any ``new'' BFS tokens; see Lemma~\ref{lemma:disjoint_paths}).

The problem with this argument is that if different neighbors of $u_0$ \emph{share} paths to $u_i$,
we might be overcounting when we charge each path against the degree of $u_i$.
Our solution is to show that there is ``not too much'' sharing, otherwise a $2k$-cycle appears --- and since we only consider neighbors of $u_0$ that are not on a $2k$-cycle, we know that this is impossible.

In Figure~\ref{fig:bad_sharing}, we show an example of one situation that must be ruled out (among others): consider $k = 5$ (i.e., 10-cycles),
and suppose that
two distinct neighbors $s, s' \in N(u_0)$ each have at least 10 node-disjoint paths with the ``right colors'', 0-1, to $u_2$.
Suppose further that one of these paths is \emph{shared}, as shown in the figure.
In addition, node $s'$ has at least one additional path (the rightmost path in the figure),
which must exist because $s'$ has at least 10 node-disjoint paths to $u_2$ (so at least one of these paths
avoids all the other nodes shown in the figure).
We see that there is a 10-cycle involving nodes $s$ and $s'$;
since we only consider neighbors of $u_0$ that do not themselves participate in a 10-cycle, this situation cannot arise.

\subsubsection{Analysis of the Meta-Algorithm}
Since nodes reject only if they \emph{find} a copy of $C_{2k}$ (by having their BFS token return to them in $2k$ color-coded steps,
or by receiving the BFS of some node at distance $k$ through two node-disjoint paths),
if the graph contains no copy of $C_{2k}$, then all nodes accept.
We therefore focus on the case where the graph does contain a heavy copy of $C_{2k}$, and show that the meta-algorithm can find it.

\begin{lemma}
	If the graph contains a heavy $2k$-cycle, then with probability $9/10$, some node rejects.
	\label{lemma:heavy_reject}
\end{lemma}

To prove Lemma~\ref{lemma:heavy_reject}, we show that for each $k = 2,3,4,5$,
there is a choice of $T_k : \set{1,\ldots,2k-1} \rightarrow \nat$ such that one iteration of the meta-algorithm detects a heavy copy of $C_{2k}$, if there is one, with probability $1/O(n^{1-1/k})$ (treating $k$ as a constant).
Therefore, after $R = \Theta(n^{1-1/k})$ iterations, we reject with high probability.

Let $u_0,\ldots,u_{2k-1}$ be a heavy cycle, and assume that $u_0$ is a node with the largest degree in the cycle (i.e., $\deg(u_0) \geq \deg(u_i)$ for each $i \in \set{1,\ldots,2k-1}$).
In particular, since the cycle is heavy, we have $\deg(u_0) > n^{1/k}$.
We consider two cases:
\begin{enumerate}
	\item Node $u_0$ has at least $n^{1/k}/100$ neighbors that each belong to some $2k$-cycle.
		In this case, when we sample a uniformly random node $s \in V$, we have probability at least $(n^{1/k}/100)/n = 1/(100n^{1-1/k})$ that $s$ is on a $2k$-cycle; and given that $s$ is indeed on a $2k$-cycle, we will find the $2k$-cycle with probability $99/100$ after $R$ iterations of color-BFS (provided we choose a large enough constant in $R$).
		Therefore, in this case, we reject with probability $\Omega(1/n^{1-1/k})$.
	\item Node $u_0$ has at least $(99/100)n^{1/k}$ neighbors that do not belong to any $2k$-cycle.
		We consider the following event $\mathcal{E}_k$:
		\begin{enumerate}
			\item $s \in N(u_0)$, and
			\item $s$ is not on any $2k$-cycle, and
			\item $c({u_i}) = i$ for each $i \in [2k]$, and
			\item $s \not \in B_k(u_0)$, where $B_k(u_0)$ is a set of ``bad neighbors'' of node $u_0$,
				which is defined in a different way for each $k$.
		\end{enumerate}
\end{enumerate}

Next, we consider each $k = 2,3,4,5$ separately,
define $T_k$ and $B_k(u_0)$, and prove that
\begin{enumerate}
	\item The number of ``bad neighbors'' $B_k(u_0)$ that are not on any $2k$-cycle is bounded from above by $\alpha_k \cdot \deg(u_0)$,
		where $\alpha_k \in (0,99/100)$ is some constant fraction that depends only on $k$.
Therefore, $u_0$ has $\Omega(n^{1/k})$ neighbors that are not in $B_k(u_0)$ and are also not on any $2k$-cycle.
The probability of hitting such a neighbor is $\Omega(1/n^{1-1/k})$.
Independent of this event, the probability that the cycle $u_0,\ldots,u_{2k-1}$ is colored correctly is constant,
and therefore $\mathcal{E}_k$ occurs with probability $1/O(n^{1-1/k})$.
	\item 
Conditioned on $\mathcal{E}_k$,
each cycle node $u_i$ for $i \in [2k]\setminus\set{0,k}$ receives no more than $T_k(i)$ distinct BFS tokens.
This means that conditioned on $\mathcal{E}_k$, the color-BFS completes successfully, causing node $u_k$ to reject.
\end{enumerate}
Together we see that we have probability $1/O(n^{1-1/k})$ of detecting $u_0,\ldots,u_{2k-1}$ in each of the $R'$
iterations, as desired.

This part of the analysis proceeds as follows.
We say that a neighbor $s \in N(u_0)$ is \emph{free} if $s$ does not participate in a $2k$-cycle.
Let $N'(u_0)$ denote the free neighbors of $u_0$,
and let $\deg'(u_0) = |N'(u_0)|$.

After choosing a neighbor $s \in N(u_0)$, we 
check if $s$ participates in a $2k$-cycle,
and if not,
we initiate a BFS from every 0-colored neighbor of $s$.
We must show that not too many BFS tokens --- at most $T_k(i)$ --- can reach a given cycle node $u_{ b \cdot i}$ where $b \in \set{-1,+1}$ and $i \in \set{1,\ldots,k-1}$.
Thus, we want to show that the ``typical'' free neighbor $s \in N'(u_0)$
does not have many disjoint paths of length $i+1$ to $u_i$, through which BFS tokens can flow to $u_i$.

Given $b \in \set{-1,+1}, i \in \set{1,\ldots,k-1}$,
we say that a path 
$\pi = w_0, \ldots,w_{i-1}$ is an \emph{$(i,b)$-path of $s$} if 
\begin{enumerate}
	\item $w_0 \in N(s)$ and $w_{i-1} \in N(u_{b \cdot i})$,
	\item The path has ``the right colors'' so that node $w_0$ initiates a BFS that flows across the path and reaches $u_{b \cdot i}$:
		that is, $c(w_j) = b \cdot j$ for each $j = 0,\ldots,i-1$.
	\item The path is node-disjoint from the prefix $u_0, u_b, \ldots, u_{b\cdot (i-1)}$ of the cycle.
\end{enumerate}

In the sequel, to simplify the presentation, we consider $b = 1$; the case $b = -1$ is symmetric.
We simplify our notation by writing ``$i$-path'' instead of ``$(1,i)$-path''.

Our goal is to show that a large fraction of free neighbors $s \in N'(u_0)$ have only a small number of node-disjoint $i$-paths, for each $i = 0,\ldots,k-1$,
as this ensures that congestion is well-controlled:
\begin{lemma}
	Suppose we have sampled a neighbor $s \in N'(u_0)$ which has no more than $p$ node-disjoint $i$-paths.
Then the number of BFS tokens that arrive at cycle node $u_{i}$ is bounded by $(p+1) \left[ \sum_{j = 1}^{i-1} T_k(j) \right]$.
	\label{lemma:disjoint_paths}
\end{lemma}
\begin{proof}
	Let $\pi_1 = (\pi_1^0,\ldots,\pi_1^{i-1}),\ldots,\pi_p = (\pi_p^0,\ldots,\pi_p^{i-1})$
	be a maximal set of node-disjoint $i$ paths from $s$ to $u_{i}$.
	Suppose for the sake of contradiction that $u_{i}$ receives $t > (p+1) \left[ \sum_{j = 1}^{i-1} T_k(j) \right]$ BFS tokens.

	Note that for each $j = 1,\ldots,i-1$, nodes $\pi_1^j,\ldots,\pi_p^j$ each forward at most $T_k(j)$ tokens.
	In particular, since the last node of each path $\pi_1,\ldots,\pi_p$ forwards at most $T_k(i-1)$ tokens,
	and node $u_{i-1}$ also forwards at most $T_k(i-1)$ tokens,
	we have at least $t - (p+1) T_k(i-1) > (p+1) \left[ \sum_{j = 1}^{i-2} T_k(j) \right]$ BFS tokens
	that arrived at $u_{i}$ without passing through $u_{i-1}$ or through any of the nodes $\pi_1^{i-1},\ldots,\pi_p^{i-1}$.
	Next, since each next-to-last node on $\pi_1,\ldots,\pi_p$, as well as $u_{i-2}$, each forward at most $T_k(i-2)$ tokens,
	we have at least $t - (p+1) \left[ T_k(i-1) + T_k(i-2) \right] > (p+1) \left[ \sum_{j = 1}^{i-3} T_k(j) \right]$ BFS tokens
	that arrived at $u_{i}$ without passing through the last two nodes on any path $\pi_1,\ldots,\pi_p$,
	or through $u_{i-2}, u_{i-1}$.
	Continuing in a similar manner, we eventually see that there must be at least 
	$t - (p+1) \left[ \sum_{j = 1}^{i-1} T_k(j) \right] > 0$ tokens --- i.e., at least one token --- that arrived at $u_{i}$ without passing through $u_0,\ldots,u_{i-1}$
	or through any of the nodes on paths $\pi_1,\ldots,\pi_p$;
	this token must have been forwarded along some path $\tau = \tau^0,\ldots,\tau^{i-1}$,
	where $\tau^0 \in N(s)$,
	$\tau^{i-1} \in N(u_{i})$, and $c(\tau^j) = j$ for each $j = 0,\ldots,i-1$,
	and $\tau$ is node-disjoint from all the paths $\pi_1,\ldots,\pi_p$.
	Note that $\tau$ is ``colored correctly'', otherwise a BFS token could not flow across it;
	so $\tau$ is in fact an $i$-path.
	This contradicts our assumption that $\pi_1,\ldots,\pi_p$
	is a maximal set of node-disjoint $i$-paths from $s$ to $u_{i}$.
\end{proof}

For each $i = 1,\ldots,k-1$ and $b \in \set{-1,+1}$,
define 
\begin{equation*}
	B_k^{b,i}(u_0) = \set{ s \in N(u_0) \stt \text{$s$ is free and has at least $d_k$ node-disjoint $(b,i)$-paths}}
\end{equation*}
to be the ``bad neighbors'' of $u_0$, where here $d_k \in \mathbb{N}$ is some constant (which depends on $k$).
We prove that there are not too many bad neighbors:
\begin{equation}
	\sum_{i = 1}^{k-1} | B_k^{b,i}(u_0) | < \alpha \deg(u_i),
	\label{eq:good1}
\end{equation}
where $\alpha < 1/4$ is some constant.
Since we assume that $\deg(u_0) \geq \deg(u_i)$ and that $\deg'(u_0) \geq \deg(u_0) / 2$ (in this part of the analysis),
and accounting for both $b = +1$ and $b = -1$,
we have
\begin{equation}
	\left|
		N'(u_0)
		\setminus 
		\bigcup_{b \in \set{-1,+1}} \bigcup_{i = 1}^{k-1} B_k^{b,i}(u_0)
	\right|
	> (1 - 4\alpha) \deg'(u_0) = \Omega(n^{1/k}).
	\label{eq:good}
\end{equation}
By Lemma~\ref{lemma:disjoint_paths}, when we sample a good neighbor, the cycle nodes do not have too much congestion,
and the BFS token of $u_0$ is able to reach $u_{k-1}$ and $u_{k+1}$.

\subparagraph*{Controlling the number of bad neighbors.}
Let us again assume $b = +1$ and drop $b$ from our notation.

To prove~\eqref{eq:good1}, we observe that 
any bad neighbor $s \in B_k^i(u_0)$ contributes at least $d_k$ to the degree of $u_i$, as $s$ has at least $d_k$ node-disjoint $i$-paths
which connect to $u_i$ through $d_k$ different neighbors of $u_i$.
Unfortunately, it could be that two different bad neighbors $s, s' \in B_k^i(u_0)$ \emph{share} some of their $i$-paths,
so we cannot immediately argue that the number of bad paths is bounded by $\deg(u_i) / d_k$.
The bulk of the proof consists of showing that ``not too many'' bad neighbors can share ``too many'' of their $i$-paths,
and therefore we can still show that the number of bad neighbors is $O(\deg(u_i))$.
Indeed, we show that ``too much sharing'' of $i$-paths between different bad neighbors creates a $2k$-cycle through them,
and since we only consider \emph{free} neighbors of $u_0$, this cannot happen.

We proceed to consider each $k = 2,3,4,5$ separately.

\paragraph*{Analysis for $k = 2$ (i.e., 4-cycles).}
We set $T_2(1) = 1$ and $B_2(u_0) = \emptyset$.

Suppose for the sake of contradiction that node $u_b$, where $b \in \set{-1,+1}$, receives more than one BFS token.
Then there is some node $v \neq u_0$, whose BFS token $u_b$ received;
both $u_0$ and $v$ are neighbors of $u_b$.
In addition, since we only start a BFS from neighbors of $s$, node $v$ must be a neighbor of $s$.
Thus, the graph contains a $4$-cycle that includes $s$: $u_0, u_b, v, s$.
This contradicts our assumption that $s$ is not on a $4$-cycle.

\paragraph*{Analysis for $k = 3$ (i.e., 6-cycles).}
We set $T_3(1) = T_3(2) = 3$, and define ``bad neighbors'' as follows:
\begin{equation*}
	B_3(u_0) = \set{ v \in N(u_0) \stt |N(v) \cap N(u_1)| > 3 \text{ or } |N(v) \cap N(u_{-1})| > 3}.
\end{equation*}

First, observe that $u_0$ has at most two bad neighbors that are not on any $6$-cycle:
we show that there is at most one node $v \in N(u_0)$ which is not on any $6$-cycle and has $|N(v) \cap N(u_1)| > 3$, and similarly when we replace $u_1$ with $u_{-1} = u_5$.
Suppose for the sake of contradiction that there are two nodes $v \neq v'$ such that $v, v' \in N(u_0)$, neither $v$ nor $v'$ are on a $6$-cycle,
and also $|N(v) \cap N(u_1)| > 3, |N(v') \cap N(u_1)| > 3$.
Then there exist nodes $w \in \left( N(v) \cap N(u_1) \right) \setminus \set{ u_0, v' }$,
$w' \in \left( N(v') \cap N(u_1) \right) \setminus \set{ u_0, v, w}$ (because after removing at most 3 nodes from $N(v) \cap N(u_1)$ or from $N(v') \cap N(u_1)$, the sets are still not empty).
Since we assume the graph contains no self-loops, we also have $w \neq v, u_1$ and $w' \neq v', u_1$, as $v, v'$ are not neighbors of themselves.
Therefore the following 6-cycle is in the graph: $v, w, u_1, w', v', u_0$.

Next we show that conditioned on $\mathcal{E}_3$, each cycle node $u_i$ or $u_{-i}$ where $i \in \set{1,2}$ receives at most $T_3(|i|) = 3$ BFS tokens.
We prove it for $u_1$ and $u_2$; the proof for $u_{-1}$ and $u_{-2}$ (resp.) is similar.
\begin{itemize}
	\item $u_1$: since $\mathcal{E}_3$ requires that $c(u_1) = 1$, node $u_1$ only receives BFS tokens in the first step of the color-BFS; that is, $u_1$ only receives BFS tokens from its own neighbors which are also neighbors of $s$ (and are colored 0). Because $s \not \in B_3(u_0)$ under $\mathcal{E}_3$, there are at most three such BFS tokens.
	\item $u_2$: since $\mathcal{E}_3$ requires that $c(u_2) = 2$, node $u_2$ only receives BFS tokens in the second step of the color-BFS.
		We already showed that $u_1$ receives at most three BFS tokens; thus, in order for $u_2$ to receive more than three,
		the fourth token must come through some node other than $u_1$.

		If $u_1$ receives no more than three BFS tokens, these include the BFS token of $u_0$,
		so the fourth token received by $u_2$ cannot originate at $u_0$.
		Therefore 
		there must exist $v \neq u_0$ and $w \neq u_1$ such that
		\begin{itemize}
			\item $v \neq u_0$ is the originator of the fourth BFS token received by $u_2$:
				we have
				$v \in N(s)$ (and $c(v) = 0$, but we do not need this fact).
			\item $w \neq u_1$ is the node that forwards $v$'s token to $u_2$: we have $w \in N(v) \cap N(u_2)$.
		\end{itemize}
		However, this means that $s$ has two node-disjoint paths of length two to $u_2$, so it participates in the following $6$-cycle:
		$s, u_0, u_1, u_2, w, v$.
		Under $\mathcal{E}_3$ we know that $s$ does not participate in any $6$-cycle, so this is impossible.

		On the other hand, if $u_1$ receives more than three BFS tokens, it forwards no tokens to $u_2$.
		Of the four (or more) tokens received by $u_2$, at least one belongs to some node $v \neq u_0$.
		So again, we have nodes $v \neq u_0$ and $w \neq u_1$ such that $v \in N(s)$ and $w \in N(v) \cap N(u_2)$,
		and we get a 6-cycle that includes $s$, as above.
\end{itemize}

\paragraph*{Analysis for $k = 4$ (i.e., 8-cycles).}
We say that a node $s$ is \emph{free} if it does not participate in any 8-cycle.
Our analysis considers two cases, depending on whether or not a certain pattern is present in the graph.

With respect to the fixed cycle $u_0,\ldots,u_7$,
and given $b \in \set{-1,+1}$,
we define a \emph{$b$-pattern} $D$ to be the following 4-node subgraph, which is node-disjoint from the cycle:
$D = (\set{ s, s', w, w'}, \set{ \set{ s, w}, \set{s, w'}, \set{ s', w'}})$, such that 
in addition to the internal edges of $D$, we have
\begin{itemize}
	\item $s, s' \in N(u_0)$,
	\item $w, w' \in N(u_b)$.
\end{itemize}
Nodes $s, s'$ are called the \emph{heads} of the dangerous pattern,
and $w, w'$ are called the \emph{tails}.

\begin{observation}
	If there is a $b$-pattern $D$ with heads $s, s'$, at least one of which is free,
	and with tails $w, w'$,
	then there cannot exist any free node $s'' \in N(u_0) \setminus \set{ s, s', w, w'}$ such that $N(s'') \cap N(u_b) \not \subseteq \set{ w, w', u_0, s, s'}$.
	\label{obs:dangerous}
\end{observation}
\begin{proof}
	Suppose otherwise,
	and let $w'' \in N(s'') \cap N(u_b) \setminus \set{ w, w', u_0, s, s'}$.
	Then the following 8-cycle is in the graph:
	$s, w', s', u_0, s'', w'', u_b, w$, contradicting our assumption that at least one of the nodes $s, s'$ is free (i.e., does not participate in an 8-cycle).

	We verify that this is indeed a simple 8-cycle:
	\begin{itemize}
		\item $w' \neq s$ because they are distinct nodes of $D$,
		\item $s' \not \in \set{s, w'}$ for the same reason,
		\item $u_0 \not \in \set{s, w', s'}$ because $D$ is disjoint from the cycle,
		\item $s'' \not \in \set{s, w', s', u_0}$ because we assumed that $s'' \in N(u_0) \setminus \set{s, s', w, w'}$ and the graph contains no self-loops,
		\item $w'' \not \in \set{s, w', s', u_0, s''}$ by choice of $w''$, together with the fact that $w'' \in N(s'')$ and the graph contains no self-loops,
		\item $u_b \not \in \set{s, w', s', u_0, s'', w''}$: we know that $w, w', w'' \in N(u_b)$, and the graph contains no self-loops;
			we cannot have $u_b = u_0$ because these are distinct nodes of our fixed 8-cycle;
			and we cannot have $u_b \in \set{s, s'}$ because $D$ is node-disjoint from the cycle.
		\item $w \not \in \set{s, w', s', u_0, s'', w'', u_b}$: we know that $w \not \in \set{ s, s', w'}$
			because these are distinct nodes of $D$;
			also, $w \not \in \set{ u_0, u_b}$ because $D$ is distinct from the 8-cycle;
			and finally, $w \not \in \set{ s'', w''}$ by choice of $s'', w''$.
	\end{itemize}
\end{proof}

The set of bad neighbors, $B_4(u_0)$, is defined as follows:
\begin{enumerate}[(I)]
	\item Define a \emph{0-1 $b$-path} from node $v \in N(u_0)$ to node $u_{2b}$ to be a path of length 2 between these nodes, $\pi = v, w_0, w_1, u_{2b}$,
		such that $c(w_0) = 0, c(w_1) = 1, w_0 \neq u_0, w_1 \neq u_b$.
		Two 0-1 paths $\pi_1 = v, w_0^1, w_1^1, u_{2b}$ and $\pi_2 = v, w_0^2, w_1^2, u_{2b}$ are called \emph{node-disjoint}
		if $w_0^1 \neq w_0^2$ and $w_1^1 \neq w_2^2$ (note that because of differing colors, node-disjoint 0-1 paths cannot share any nodes,
		except the two endpoints $v, u_{2b}$).

		Any neighbor $v \in N(u_0)$ that has at least four node-disjoint 0-1 paths to $u_{2b}$ is added to $B_4(u_0)$.
	\item For each $b \in \set{-1,+1}$, if the graph contains a $b$-pattern w.r.t.\ $u_0,\ldots,u_7$, we fix one such pattern arbitrarily, and add its heads to $B_4(u_0)$.
	\item For each $b \in \set{-1,+1}$, if the graph does not contain a $b$-pattern w.r.t.\ $u_0,\ldots,u_7$, then we add to $B_4(u_0)$ any node $v$ with $|N(v) \cap N(u_b)| \geq 4$ for $b \in \set{-1,+1}$.
\end{enumerate}

First, we bound the number of free bad neighbors of $u_0$ of each type I-III, and show that $|B_4(u_0)| \leq (3/4)\deg(u_0)$:
\begin{enumerate}
	\item For each $b \in \set{-1,+1}$, there is at most one free node $v \in N(u_0)$ that has four node-disjoint 0-1 paths to $u_{2b}$:
		suppose for the sake of contradiction that there are two such nodes, $v \neq v'$.
		Then we have the following paths in the graph:
		\begin{itemize}
			\item $v, u_0, u_b, u_{2b}$,
			\item A 0-1 path, $v, w_0, w_1, u_{2b}$, which is node-disjoint from the previous path by definition,
			\item A path $v', w_0', w_1', u_{2b}$ which is node-disjoint from the previous paths (such a path exists because $v'$ has at least four node-disjoint 0-1 paths to $u_{2b}$, none of which include $u_0$ or $u_b$; at least one of these paths avoids $v, w_0, w_1$).
		\end{itemize}
		Therefore, the graph includes the following 8-cycle:
		$v, u_0, v', w_0', w_1', u_{2b}, w_1, w_0$,
		contradicting our assumption that $v, v'$ are free.
	\item For each $b \in \set{-1,+1}$, if the graph contains a $b$-pattern $D$, then it has exactly two heads, so we add two nodes to $B_4(u_0)$.
	\item For each $b \in \set{-1,+1}$, if the graph does not contain a $b$-pattern, then for any two neighbors $s, s'\in N(u_0)$, if either $|N(s) \cap N(u_b)| \geq 4$ or $|N(s') \cap N(u_b)| \geq 4$, then we must have $N(s) \cap N(s') \cap N(u_b) = \set{u_0}$: otherwise, if w.l.o.g.\ we had $|N(s) \cap N(u_b)| \geq 4$ and also $N(s) \cap N(s') \cap N(u_b) \supsetneq \set{u_0}$, then there would exist tails, $w \in N(s) \cap N(s') \cap N(u_b) \setminus \set{u_0}$ and $w' \in N(s) \cap N(u_b) \setminus \set{ u_0, w, s}$, such that $s, s', w, w'$ are a $b$-pattern w.r.t.\ $u_0,\ldots,u_7$.

		Let $U$ be the set of nodes $s$ with $|N(s) \cap N(u_b)| \geq 4$.
		As we just said, for any distinct $s, s' \in U$,
		we have $N(s) \cap N(s') \cap N(u_b) = \set{u_0}$,
		that is, $\left( N(s) \cap N(u_b) \right) \cap \left( N(s') \cap N(u_b) \right) = \set{u_0}$.
		Since we assumed that $u_0$ has maximal degree among $u_0,\ldots,u_7$,
		\begin{equation*}
			\left|N(u_0)\right| \geq \left| N(u_b) \right| \geq \left| \bigcup_{s \in U} N(s) \cap N(u_b) \right|
			\geq 1 + 3 \cdot |U|.
		\end{equation*}
		We see that
		$|U| < \deg(u_0) / 3$.
\end{enumerate}
Summing across both $b = -1,+1$, we see that the total number of bad neighbors is bounded by $6+2\deg(u_0)/3 < (3/4)\deg(u_0)$, assuming $n$ is large enough (recall that $\deg(u_0) \geq n^{1/k}$, so for $n$ large enough we have $\deg(u_0)/12 > 6$).

Next, assume we have sampled a free good neighbor $s \in N(u_0) \setminus B_4(u_0)$, and let us bound the number of BFS tokens that
each cycle node can receive.
Let $b \in \set{-1,+1}$.
\begin{itemize}
	\item $u_b$ can receive at most $5$ BFS tokens: since we assume that $c(u_b) = b$,
		the only BFS tokens received by $u_b$ are those sent by 0-colored nodes in $N(s) \cap N(u_b)$.
		We consider two cases:
		\begin{enumerate}
			\item The graph contains a $b$-pattern with heads $v, v'$ and tails $w, w'$:
				then by definition, since $s$ is not a bad neighbor, $s \not \in \set{v, v'}$.
				By Observation~\ref{obs:dangerous},
				we have $N(s) \cap N(u_b) \subseteq \set{ v, v', w, w', u_0}$,
				so at most 5 BFS tokens can reach $u_b$.
			\item The graph does not contain a $b$-pattern: then since $s$ is not a bad neighbor,
				we have $|N(s) \cap N(u_b)| < 4$, and hence fewer than 5 BFS tokens can reach $u_b$.
		\end{enumerate}
	\item $u_{2b}$ can receive at most 30 BFS tokens:
		since $s$ is not bad, it has at most four node-disjoint 0-1 paths to $u_{2b}$.
		Let $\pi_1,\ldots,\pi_{\ell}$, $\ell \leq 4$, be a maximal set of node-disjoint 0-1 paths
		from $s$ to $u_{2b}$.
		For each such path $\pi_i = s, w_0^i, w_1^i, u_{2b}$, if node $w_1^i$ receives more than 5 BFS tokens, it sends none of them;
		and if it receives at most 5 BFS tokens, it forwards them to $u_{2b}$.
		The same goes for the path $s, u_0, u_b, u_{2b}$.
		Thus, node $u_{2b}$ receives at most 25 tokens from nodes $\set{w_1^i}_{i = 1,\ldots,\ell}$ and $u_b$.
		We also ``throw in for free'' the BFS tokens of nodes $\set{ w_0^i}_{i = 1,\ldots,\ell}$ and $u_0$,
		for a total of at most 30 tokens received at $u_{2b}$.
		(These latter tokens may reach $u_{2b}$ through some node other than $\set{w_1^i}_{i=1,\ldots,\ell}, u_b$, and we pessimistically assume that they do.)

		Suppose for the sake of contradiction that $u_{2b}$ receives more than 30 tokens.
		Then one of these tokens was neither originated by one of the nodes $\set{ w_0^i}_{i = 1,\ldots,\ell}, u_0$,
		nor forwarded by one of the nodes $\set{w_1^i}_{i = 1,\ldots,\ell}, u_b$.
		This means that there is some 0-colored neighbor $x \in N(s)$, such that $x \not \in \set{ w_0^i}_{i=1,\ldots,\ell} \cup \set{u_0}$,
		whose token was received by $u_{2b}$,
		and a 1-colored neighbor $y \in N(s) \cap N(u_{2b})$, such that $y \not \in \set{ w_1^i}_{i = 1,\ldots,\ell} \cup \set{u_b}$,
		that forwarded $x$'s token to $u_{2b}$.
		But then the path $s, x, y, u_{2b}$ is a 0-1 path that is node-disjoint from $\pi_1,\ldots,\pi_{\ell}$,
		contradicting our assumption that this is a maximal set of node-disjoint 0-1 paths from $s$ to $u_{2b}$.
		
	\item $u_{3b}$ can receive at most 36 BFS tokens: suppose for the sake of contradiction that $u_{3b}$ receives more than 36 tokens.
		Since $u_0$ originates one token, and nodes $u_{b}, u_{2b}$ forward 5 tokens and 30 tokens, respectively,
		this means that node $u_{3b}$ receives some token originated by a neighbor $w_0 \neq u_0$ of $s$,
		and forwarded first by $w_1 \neq u_b$ and then by $w_2 \neq u_{2b}$.
		Therefore the graph contains the 8-cycle $s, u_0, u_b, u_{2b}, u_{3b}, w_2, w_1, w_0$, contradicting our assumption that $s$ is free.

\end{itemize}

\paragraph*{Analysis for $k = 5$ (i.e., 10-cycles).}
Let $B_5^1(u_0)$ be the set of free neighbors of $u_0$ that have 100 or more different 1-paths to $u_1$.
(Recall that a 1-path is simply one node $w_0$, colored 0, and connected to both $s$ and $u_1$.)

\begin{lemma}
	Suppose nodes $s, s' \in B_5^1(u_0)$ ($s \neq s'$) have a common 1-path,
	$w_0 \in N(s) \cap N(s') \cap N(u_1)$.
	Then for any two other nodes $s'', s''' \in B_5^1(u_0) \setminus \set{s, s', w_0}$ ($s'' \neq s'''$),
	there is no common 1-path $w_0' \in N(s'') \cap N(s''') \cap N(u_1) \setminus \set{ s, s', w_0, u_0 }$.
	\label{lemma:C10_u1}
\end{lemma}
\begin{proof}
	Suppose the lemma is false, and let $s, s', s'', s''', w_0, w_0'$ be as in the lemma.
	Since $s \in B_5^1(u_0)$, it has at least 100 1-paths to $u_1$,
	and at least one of them, call it $x_0$, excludes nodes $s, s'', s''', w_0, w_0', u_1$.
	Also, since $s''' \in B_5^1(u_0)$, it has at least one 1-path, call it $y_0$,
	which differs from $s, s', s'', w_0, w_0', u_1, x_0$.
	Therefore the following 10-cycle is in the graph:
	$s, w_0, s', u_0, s'', w_0', s''', y_0, u_1, x_0$.
	This contradicts our assumption that $s$ (and also $s', s'', s'''$) are free.
\end{proof}

\begin{corollary}
	Assuming $\deg'(u_0) > \deg(u_0) / 2 > 100$,
	we have $|B_5^1(0)| \leq \deg'(u_0)/20$.
	\label{cor:C10_u1}
\end{corollary}
\begin{proof}
	We claim that $|B_5^1(u_0)| \leq \deg(u_2)/50+3$.
	Since $\deg(u_2) \leq \deg(u_0) \leq 2\deg'(u_0)$,
	this implies that 
	\begin{equation*}
		|B_5^1(u_0)| \leq \frac{2\deg'(u_0)}{50}+3 < \frac{\deg'(u_0)}{20}.
	\end{equation*}

	If no two nodes $s \neq s' \in B_5^1(u_0)$ have a common 1-path
	$w_0 \in N(s) \cap N(s') \cap N(u_1)$,
	then each node in $B_5^1(u_0)$ contributes at least 100 unique neighbors of $u_1$
	which are not contributed by any other neighbor in $B_5^1(u_0)$,
	and therefore $\deg(u_2) \geq 100|B_5^1(u_0)|$.
	Thus, assume there do exist $s \neq s' \in B_5^1(u_0)$ with a common 1-path $w_0$,
	and fix such $s, s', w_0$.
	The remaining nodes in $B_5^1(u_0) \setminus \set{s, s', w_0}$
	do not have any common 1-paths among themselves,
	except possibly $s, s', w_0, u_0$;
	but each node in $B_5^1(u_0)$ has at least 100 1-paths to $u_1$,
	and at least 50 of them are not $s, s', w_0, u_0$, and as we just said,
	are therefore not shared with any other node in $B_5^1(u_0) \setminus \set{ s, s', w_0}$.
	It follows that each node in $B_5^1(u_0) \setminus \set{ s, s', w_0}$
	contributes at least 50 unique neighbors of $u_2$,
	and hence
	$|B_5^1(u_0)| \leq \deg(u_2)/50 + 3$.

\end{proof}

Let $B_5^2(u_0)$ be the set of free neighbors of $u_0$
that have 100 or more node-disjoint 2-paths to $u_2$.

\begin{lemma}
	Assuming $\deg'(u_0) > \deg(u_0) / 2 > 1000$,
	we have $|B_5^2(u_0)| \leq \deg'(u_0)/10$.
	\label{lemma:C10_u2}
\end{lemma}
\begin{proof}
	Suppose for the sake of contradiction that $|B_5^2(u_0)| > \deg'(u_0)/10 > 100$.
	We claim that no two nodes in $B_5^2(u_0)$ can \emph{share} a 2-path, that is,
	there cannot exist $s, s' \in B_5^2(u_0)$
	and 2 paths $w_0, w_1$ and $w_0', w_1'$ which are \emph{not node-disjoint},
	such that $w_0, w_1$ is a 2-path from $s$ to $u_2$,
	and $w_0', w_1'$ is a 2-path from $s'$ to $u_2$.

	Suppose there exist such $s, s'$ and paths such that $\set{ w_0, w_1} \cap \set{ w_0', w_1'} \neq \emptyset$.
	Since $c(w_0) = c(w_0') = 0$ and $c(w_1) = c(w_1') = 1$,
	either $w_0 = w_0'$ or $w_1 = w_1'$.
	\begin{itemize}
		\item If $w_0 = w_0'$, then there cannot exist any $s'' \in B_5^2(u_0) \setminus \set{s, s', w_0, w_0', w_1, w_1', u_0, u_1, u_2}$,
			contradicting our assumption about the size of $B_5^2(u_0)$:
			if $s''$ exists, then since it has at least 100 node-disjoint 2-paths to $u_2$,
			at least one of these paths, call it $w_0'', w_1''$, excludes nodes $s, s', w_0, w_0', w_1, w_1', u_0, u_1, u_2$.
			In addition, since $s' \in B_5^2(u_0)$, it also has at least one additional 2-path to $u_2$, call it $x_0, x_1$,
			which excludes nodes $s, s'', w_0, w_0', w_1, w_1', u_0, u_1, u_2, w_0'', w_1''$.
			We therefore have the following 10-cycle:
			$s'', w_0'', w_1'', u_2, x_1, x_0, s', w_0' = w_0, s, u_0$.
			This contradicts our assumption that $s, s', s''$ are free neighbors of $u_0$.
		\item If $w_0 \neq w_0'$ but $w_1 = w_1'$: since $s' \in B_5^2(u_0)$, it has at least one additional 2-path to $u_2$, call it $x_0, x_1$,
			which excludes nodes $\set{ s, w_0, w_0', w_1 = w_1', u_0, u_1, u_2}$.
			Therefore the following 10-cycle is in the graph:
			$s, w_0, w_1 = w_1', w_0', s', x_0, x_1, u_2, u_1, u_0$.
			Again, this contradicts our assumption that $s, s'$ are free neighbors of $u_0$.
	\end{itemize}
	We see that each $s \in B_5^2(u_0)$ contributes at least 100 2-paths to $u_2$,
	which are node-disjoint from the 2-paths contributed by any other node in $B_5^2(u_0)$,
	and therefore we can charge each node in $B_5^2(u_0)$ with 100 $1$-colored vertices in the neighborhood of $u_2$
	(which are not double-charged to any other node in $B_5^2(u_0)$).
	It follows that $\deg(u_2) \geq 100B_5^2(u_0)$.
	Since we assume that $\deg(u_0) \geq \deg(u_2)$ and that $\deg'(u_0) > \deg(u_0) / 2$,
	we get that $|B_5^2(u_0)| \leq \deg(u_2) / 100 \leq \deg(u_0) / 100 < \deg'(u_0) / 50$,
	a contradiction to our assumption that $B^2_5(u_0)$ is large.
\end{proof}

Let $B_5^3(u_0)$ be the set of free neighbors of $u_0$
that have at least 10 node-disjoint 3-paths to $u_3$.
\begin{lemma}
	We have $|B_5^3(u_0)| \leq 1$.
	\label{lemma:C10_u3}
\end{lemma}
\begin{proof}
	Suppose not, and let $s, s' \in B_5^3(u_0)$ be distinct nodes.
	Let $w_0, w_1, w_2$ be a 3-path of $s$ to $u_3$,
	and let $w_0', w_1', w_2'$ be a 3-path of $s'$ to $u_3$,
	which avoids nodes $s, w_0, w_1, w_2, u_0, u_1, u_2$ (such a path exists, since every node in $B_5^3(u_0)$
	has at least 10 node-disjoint 3-paths to $u_3$).
	Then the following 10-cycle is in the graph:
	$s, w_0, w_1, w_2, u_3, w_2', w_1', w_0', s', u_0$.
	Therefore nodes $s, s'$ are not free, a contradiction.
\end{proof}

For any $k$, the ``last node in the proof'', $u_{k-1}$, is the easiest to handle, using the following observation:
\begin{observation}
	For any $k \geq 2$, if $u_0,\ldots,u_{2k-1}$ is a $2k$-cycle in the graph,
	and 
	$s \in N(u_0)$ is free,
	then $s$ does not have a $(k-1)$-path $w_0, \ldots,w_{k-2}$ to $u_{k-1}$ which is node-disjoint from $u_0,\ldots,u_{k-2}$.
	\label{obs:last_u}
\end{observation}
\begin{proof}
	If such a path existed, then we would have the following $2k$-cycle in the graph:
	$s, u_0,u_1, \ldots,u_{k-1}, w_{k-2},\ldots, w_0$.
	Therefore $s$ would not be a free node, contradicting our assumption.
\end{proof}

\begin{corollary}
	If $d_k \geq k$, then $B_k^{k-1} = \emptyset$.
	\label{cor:C10_u4}
\end{corollary}
\begin{proof}
	Suppose for the sake of contradiction that there is some node $s \in B_k^{k-1}(u_0)$.
	Since $s$ has at least $k-1$ node-disjoint 4-paths,
	at least one of these paths avoids nodes $u_0, u_1,\ldots,u_{k-2}$.
	By Observation~\ref{obs:last_u}, this cannot be.
\end{proof}

\subsection{Exact Algorithm for Computing the Girth in \congest}
We show that we can exactly compute the girth $g$ of a graph in time $g \cdot n^{1-1/\Theta(g)}$ in \congest.
For $g \geq \log n$, we can cap the running time at $O(n)$, because a graph with girth $\geq \log n$ has $O(n)$ edges;
thus, the running time is $O(\min \set{ g \cdot n^{1-1/\Theta(g)}, n})$.

We say that a $k$-cycle is \emph{light} if all of its nodes have degree at most $n^{\delta_k}$, where $\delta_k = k/2$ if $k$ is even, and $\delta_k = (k-1)/2$ if $k$ is odd.

The meta-algorithm is as follows:
first, we search for triangles, which can be detected in time $\tilde{O}(n^{1/3})$ using the algorithm of~\cite{CS19}.
Any node that finds a triangle outputs ``3'' for the girth.
We proceed to search for $k$ cycles for $k = 4,\ldots$:
\begin{enumerate}
	\item Search for light $k$-cycles, by simultaneously starting a depth-$\lceil k/2 \rceil$ BFS on the subgraph
		of nodes that have degree at most $n^{\delta_k}$:
		each node $u$ with $\deg(u) \leq n^{\delta_k}$ initiates a BFS, by sending a BFS token to its neighbors;
		the BFS token carries the ID of the node that originated it, and the number of hops it has traveled.
		Nodes with degree at most $n^{\delta_k}$ participate in the BFS by forwarded BFS tokens that they receive,
		increasing their hop-count, until a maximum of $\lceil k/2 \rceil$ hops
		(of course, since we are carrying out a BFS, tokens are forwarded only once).

		If node $u$ receives the BFS token of a node $v$ from two distinct neighbors of $u$, such that the total number of hops
		traveled on one side is $\lfloor k/2 \rfloor$ and on the other $\lceil k/2\rceil$, then node $u$ rejects and outputs $k$.
	\item Search for heavy $k$-cycles, by sampling a uniformly random node $s \in V$,
		\begin{enumerate}
			\item Carrying out a $k$-round BFS from $s$, to check if $s$ itself is on a $k$-cycle;
				if node $s$ receives its own BFS token back from some neighbor, it halts and outputs k.
			\item 
			Starting a depth-$\lceil k/2\rceil$ BFS from all neighbors of node $s$.
		Now, each node is allowed to forward only \emph{one} BFS token, after which it stops forwarding tokens.
		Again, if some node $u$ receives the BFS token of a node $v$ from two distinct neighbors, with $\lfloor k / 2 \rfloor$ and $\lceil k/2\rceil$
		hops traveled on the two sides (resp.),
		it halts and outputs $k$.
		\end{enumerate}
		We repeat this entire step (sampling $s$, etc.) $R = \Theta(n^{1-\delta_k})$ times.
\end{enumerate}
For a given $k$, steps (1)-(2) above are called \emph{phase $k$} of the algorithm.

Observe that if $k$ is even, then a $k$-cycle $u_0,\ldots,u_{k-1}$ is detected when node $u_{k/2}$ receives the BFS token of $u_0$ from its neighbors $u_{k/2-1}$ and $u_{k/2+1}$, with a hop count of $k/2$ on both sides;
if $k$ is odd, then a $k$-cycle $u_0,\ldots,u_{k-1}$ is detected by node $u_{(k-1)/2}$, which receives $u_0$'s token from $u_{(k-3)/2}$ and in the next round from $u_{(k+1)/2}$, with hop counts of $(k-1)/2$ and $(k+1)/2$, respectively; and simultaneously, the cycle is also detected by node $u_{(k+1)/2}$,
which receives $u_0$'s token first from $u_{(k+3)/2}$ and then from $u_{(k-1)/2}$.

\begin{lemma}
	If some node halts in phase $k$, and the graph does not contain any cycle of length less than $k$, then the graph contains a $k$-cycle.
	\label{lemma:no_false}
\end{lemma}
\begin{proof}
	Suppose node $u$ outputs $k$, after receiving 
	the token of node $v$ from two neighbors $w_1 \neq w_2$, with a hop-count of $\lfloor k/2 \rfloor$ on $w_1$'s side and $\lceil k/2 \rceil$ on $w_2$'s side.
	Then the graph contains paths $v = x_0,\ldots,x_{\lfloor k/2 \rfloor - 1}=w_1, u$
	and $v = y_0,\ldots,y_{\lceil k/2 \rceil -1}=w_2, u$,
	which together form a $k$-cycle.
	Moreover, the $k$-cycle is simple, as with the exception of $v = x_0 = y_0$ and $u$, these paths share no nodes:
	if there were some $i, j > 0$ such that $x_i = y_j$,
	then, taking the minimum such $i$ and, after fixing $i$, the minimum such $j$,
	the simple cycle $x_0,\ldots,x_i = y_i, y_{i-1},\ldots,y_0 = x_0$ would be in the graph,
	and its length would be $i + j$.
	Either $i < \lfloor k/2 \rfloor$ or $j < \lceil k/2 \rceil$ (or both),
	so $i + j < \lfloor k / 2 \rfloor + \lceil k / 2 \rceil = k$,
	but we assumed that the graph contains no cycles of length less than $k$.

	The remaining case is that in one of the iterations, a sampled node $s$ receives its own token back
	while carrying out a $k$-round BFS.
	Then $s$ participates in a $k$-cycle:
	let $s = v_0, v_1,\ldots,v_{\ell}$, $\ell \leq k-1$, be the path traveled by the token, with node $v_{\ell}$
	forwarding the token back to $s$.
	Since the graph does not contain any cycles of length less than $k$,
	we must have $\ell = k - 1$, and all nodes $v_0,\ldots,v_{k-1}$ must be distinct.
	Therefore the $k$-cycle $s = v_0,\ldots,v_{k-1}$ is in the graph.
\end{proof}

\begin{lemma}
	If we reach phase $k$, and the graph contains a $k$-cycle and has no cycles of length less than $k$,
	then with probability at least $2/3$, some node rejects in phase $k$.
	\label{lemma:reject_girth}
\end{lemma}
\begin{proof}
	Fix a cycle $u_0,\ldots,u_{k-1}$ of length $k$.
	If the cycle is light, it will be found in step (1) of the algorithm:
	since each cycle node has degree at most $n^{\delta_k}$, all these nodes participate in the BFS.

	If $k$ is even, then nodes $u_{k/2-1}$ and $u_{k/2+1}$ are able to forward the BFS token of $u_0$ to $u_{k/2}$, and 
	it arrives with hop count $k/2$ on both sides; therefore node $u_{k/2}$ rejects.
	If $k$ is odd, then node $u_{(k-3)/2}$ is able to forward the token of $u_0$ with hop count $(k-1)/2$,
	and node $u_{(k+1)/2}$ is able to forward the token of $u_0$ with hop count $k-(k+1)/2+1 = (2k-k-1+2)/2 = (k+1)/2$,
	causing node $u_{(k-1)/2}$ to reject.

	Now suppose that the cycle is heavy, and that node $u_0$ has $\deg(u_0) \geq n^{\delta_k}$.
	Then when we sample a uniformly random node $s$, with probability at least $n/\deg(u_0) \geq n^{1-\delta_k}$,
	we have $s \in N(u_0)$.
	When this occurs, we find the cycle:
	if node $s$ itself participates in a $k$-cycle, 
	then its $k$-round BFS will detect the cycle,
	because $k$ rounds suffice for the BFS token of $s$ to return to it.
	Otherwise, every neighbor of $s$, including $u_0$, starts a BFS.
	If the BFS of $u_0$ is able to traverse both paths $u_0, u_1,\ldots,u_{\lfloor k / 2 \rfloor}$ and $u_0,u_{k-1},\ldots, u_{\lfloor k/2 \lfloor}$,
	then node $u_{\lfloor k / 2 \rfloor}$ receives it with hop counts $\lfloor k / 2 \rfloor$ and $\lceil k / 2 \rceil$ respectively, and it rejects.

	Recall that in order for the BFS token of $u_0$ to traverse these paths, it must be the first token received by each node on the path.
	We show that the BFS token of $u_0$ cannot be blocked on either side, as that would imply the presence of a smaller cycle:
		suppose some node $u_{i}$ or $u_{-i}$, $i \leq \lfloor k / 2 \rfloor$, receives the BFS token of a node $w_0 \neq u_0$
			before or at the same time as it receives $u_0$'s token.
			Let $i$ be minimal, and assume w.l.o.g.\ that $u_i$ receives the token ($u_{-i}$ is symmetric, since we take $i \leq \lfloor k/2 \rfloor$).
			Then there exists a path $w_0, w_1,\ldots,w_j=u_i$ along which $w_0$'s token travels to $u_i$,
			where $w_0,\ldots,w_{j-1} \not \in \set{u_0,\ldots,u_i}$.
			Also, $w_0 \in N(s)$, since only neighbors of $s$ start a BFS.
			Therefore, the graph contains the cycle $s, w_0, \ldots,w_j=u_i, u_{i-1},\ldots,u_0$, whose length is $i+j$.
			Since $w_0$'s token arrives at $u_i$ with or before $u_0$'s token, we must have $j \leq i$.
			And since $i \leq \lfloor k/2 \rfloor$, the length of the other cycle is $i + j \leq 2\lfloor k/2 \rfloor \leq k$.
			We see that for this to occur, node $s$ must participate in a cycle of length at most $k$, but we have already ruled out this possibility.

			This shows that node $u_0$'s token is able to traverse both paths above: it cannot be blocked until it
			is forwarded by $u_{\lfloor k / 2 \rfloor - 1}$
			and $u_{\lfloor k / 2 \rfloor + 1} = u_{- \left( \lfloor k / 2 \rfloor - 1\right)} $
			to $u_{\lfloor k / 2 \rfloor}$, which then rejects.
\end{proof}

The correctness of the algorithm are implied by the following:
\begin{corollary}
	If the girth of the graph is $g$,
	then no node halts in phase $k < g$.
	Moreover, with probability at least $2/3$, 
	some node halts in phase $g$ and outputs $g$.
	\label{cor:girth_congest}
\end{corollary}
\begin{proof}
	By the first lemma, we see that when the girth is $g$, no node can halt at any phase $k < g$.
	Now consider phase $k = g$: there are no cycles of length less than $g$,
	and we already said that we do reach phase $k = g$,
	so by the second lemma, with probability at least $2/3$, some node halts.
\end{proof}

The running time of the algorithm is characterized as follows:
with probability at least $2/3$, after $g \cdot n^{1-\delta_g} = g \cdot n^{1 - 1/\lfloor g / 2 \rfloor}$ rounds,
some node halts (and outputs $g$).
However, this is not an upper bound on the \emph{expected} time until the first node halts.
If we want to bound the expected running time, we can increase the number of repetitions in each phase $k$ to $\Omega(n^{1-\delta_g} \cdot \log n)$,
so that the probability of not halting in phase $g$ is reduced to $1/n$.
The expected running time (until the first node halts) is then given by:
\begin{equation*}
	O(g \cdot n^{1 - \delta_g} \cdot \log n) \cdot \left( 1 - \frac{1}{n} \right) + O(n) \cdot \frac{1}{n} =
	O(g \cdot n^{1 - \delta_g} \cdot \log n).
\end{equation*}
(The second term uses the fact that we can cap the running time at $O(n)$ rounds, by switching to learning the entire graph if $g > \log n$.)

\subsection{Implementation of the Exact Girth and Even Cycle Algorithm in \congest}

Our implementation of the meta-algorithm for finding heavy cycles avoids sampling one node $s$ uniformly at random,
because we cannot do so in the \congest model without incurring an additive overhead of $\Omega(D)$.

Notice that in each iteration of the meta-algorithm, nodes can take on one of the following roles:
\begin{itemize}
	\item Type $\mathcal{S}$: node $s$, a unique randomly-sampled node selected in step (1) of the meta-algorithm.
	\item Type $\mathcal{NS}$: neighbors of node $s$. Each neighbor of $s$ that is colored 0 initiates a BFS in step (4) of the meta-algorithm.
	\item Type $\mathcal{O}$: all other nodes. These nodes forward BFS tokens that reach them (assuming there are not too many),
		and reject if they are colored $k$ and receive the same BFS token along two disjoint paths.
\end{itemize}
The steps taken by each node in the meta-algorithm depend only on the type it is assigned.

The implementation $\mathcal{A'}$ 
is similar to the meta-algorithm $\mathcal{A}$, but it
executes each of the $R' = \Theta(n^{1-1/k})$ iterations as follows:
\begin{enumerate}
	\item Each node $u$ chooses a random priority $p(u) \in [n^3]$. 
	\item For $2R \cdot R'$ rounds, each node $u$ forwards the smallest priority it has received so far.
		This priority is stored in the local variable $\var{pmin}(u)$.
	\item If node $u$ has $\var{pmin}(u) = p(u)$ (i.e., node $u$ has not heard any priority smaller than its own),
		it sets $\var{type}(u) = \mathcal{S}$, and informs all its neighbors.
		The neighbors $v \in N(u)$ then set $\var{type}(v) = \mathcal{NS}$.
		Nodes $u$ that do not have $\var{pmin}(u) = p(u)$ and do not have a neighbor $v$ with $\var{pmin}(v) = p(v)$
		set $\var{type}(u) = \mathcal{O}$.
	\item We now execute steps (2)-(5) of the meta-algorithm, with each node following the role it was assigned above.
\end{enumerate}

Let $\mathcal{U}$ be the event that there are no collisions in the choice of priorities, i.e., for each $u \neq v$ we have $p(u) \neq p(v)$.
This occurs with probability at least $1 - 1/n$.
Conditioned on $\mathcal{U}$, let $t$ be the uniformly-random node that has the smallest priority, $p(t) = \min_{v \in V} p(v)$,
and let $U$ be the $2R\cdot R'$-neighborhood of $t$.
Then after step (2), node $t$ is the only node in $U$ that sets $\var{type}(t) = \mathcal{S}$,
and its neighbors $v \in N(t)$ are the only nodes in $U$ that set $\var{type}(v) = \mathcal{N}$.
Thus, inside $U$, the execution of $\mathcal{A}'$ is equivalent to the execution of the meta-algorithm $\mathcal{A}$
where we sample $s = t$, a uniformly random node.
Since $\mathcal{A}$'s running time is $R \cdot R'$ rounds, this suffices to ensure that a heavy $2k$-cycle will be detected with high probability.

\section{Barrier of $\Omega(n^{1/2+\alpha})$ for Lower Bounds on $C_6$-Freeness}
\label{sec:barrier}
In this section, we show that for any $\alpha > 0$, an $\Omega(n^{1/2+\alpha})$ lower bound for $C_{6}$-freeness in CONGEST implies strong circuit complexity lower bounds. 

The proof is as follows:
\begin{enumerate}
	\item First, we reduce the problem of $C_6$-freeness to directed triangle freeness. We do so by showing that given an algorithm $\mathcal{A}_1$ for solving directed triangle freeness in $O(T_1(n))$ rounds, we can solve $C_6$-freeness in $\widetilde{O}(\sqrt{n} \cdot T_1(n))$ rounds w.h.p.;
		thus, if we can prove a lower bound on $C_6$-freeness, we also obtain a lower bound on directed triangle freeness.
	\item It is already known that proving an $\Omega(n^{\alpha})$ lower bound on \emph{undirected} triangle-freeness would imply new
		and powerful circuit lower bounds \cite{EFFKO19};
		the same argument also holds for \emph{directed} triangle-freeness. For the sake of completeness, we give the full argument in Appendix~\ref{sec:app_c6_barrier} below.
\end{enumerate}

The reduction from $C_6$ to directed triangles works as follows: the network runs $R'= O(\sqrt{n}\log{n})$ iterations of the heavy cycles finding procedure of $C_6$. By the correctness of the $C_6$-finding algorithm, if there is a $C_6$ copy with at least one node $v$ with $\deg(v) \geq \sqrt{n}$, then the network rejects with high probability. Otherwise, the network removes all nodes with $\deg(v) \geq \sqrt{n}$ from $G$.

Each node $v \in V$ proceeds to choose a random color $c(v) \in [6]$. Let $G'=(V,E')$ be the directed graph where $(u,v) \in E'$ if and only if there exists $w \in V$ such that both $(u,w),(w,v) \in E$ and $c(u)+2 \equiv c(w)+1 \equiv c(v) \pmod{6}$. 

Similarly to Section~\ref{sec:app_congest}, we say that a $6$-cycle $\set{u_0,\dots,u_5}$ of $G$ is ``colored correctly'' if $c(u_i) = i$ for each $i=0,\dots,5$. The following claim establishes that such $6$-cycles exist if and only if $G'$ has a directed triangle.

\begin{claim}\label{claim:tri}
	$G'$ contains a directed triangle if and only if the colored $G$ contains a $C_6$ copy which is colored correctly.
\end{claim}
\begin{proof}
	If $G'$ has a directed triangle $(v_1,v_2,v_3)$, then by definition the colors of its vertices hold $c(v_1)+4 \equiv c(v_2)+2 \equiv c(v_3) \pmod{6}$, and there exist $w_1,w_2,w_3$ of colors $c(v_1)+3,c(v_1)+1,c(v_1)-1 \pmod{6}$ respectively, such that $(v_1,w_1,v_2,w_2,v_3,w_3)$ is a $6$-cycle. On the other hand, if $G$ has a well colored $6$-cycle $(v_1,v_2,v_3,v_4,v_5,v_6)$, then $(v_1,v_3,v_5)$ is a triangle in $G'$.
\end{proof}

If there is a $C_6$-copy in $G$ then it becomes a correctly colored copy of $C_6$ with probability $\geq 1/6^6$. Therefore, if $G$ is not $C_6$-free, then after repeating this algorithm $O(1)$ times, with high probability at least once $G'$ has a directed triangle (from Claim \ref{claim:tri}). On the other hand, if $G$ is $C_6$-free then $G'$ never contains a directed triangle. This shows the correctness of the reduction.

We note that the network $G$ may simulate $G'$ in the following sense: first, every node $v \in V$ can learn its edges in $G'$ in $O(\sqrt{n})$ rounds:
recall that since every node with $\deg(v) \geq \sqrt{n}$ was deleted from the graph, so we can afford to have each node broadcasts all its remaining neighbors in $G$ and their colors; thus, every node learns its neighbors in $G'$.
We call a node $w$ a bridge between $u,v$ if it is a neighbor of both in $G$. We note that every $v$ also learns all its bridges to all of its neighbors in $G'$ in this process. Secondly, if the network $G$ wishes to simulate an $r$ round protocol on $G'$ where at the end each node $v \in V$ knows its output state, it can do so with $O(\sqrt{n} \cdot r)$ rounds in the following manner: for $i=1,\dots r$, let $M_i(v)$ be the set of messages $v$ wishes to send in the $i$-th round of the protocol, and for a message $m$ let $t(m)$ be the target node of that message. First, every node $v$ sends each of its neighbors $w$ in $G$ the subset of messages $\set{m \in M_i(v) \mid \text{$w$ is a bridge between $v$ and $t(m)$}}$. This can be done in $O(\sqrt{n})$ rounds as $\deg(w) \leq \sqrt{n}$ and therefore every $w$ is a bridge between $v$ and at most $O(\sqrt{n})$ other nodes. Following this, for every node $w \in V$ and message $m$ it received, $w$ sends $m$ to $t(m)$. As for any given node $t$ a node $w$ received at most $\sqrt{n}$ messages which have target $t$ (at most one from each of its neighbors), it can send these nodes their messages in $O(\sqrt{n})$ rounds. Overall, the simulation costs $O(r \cdot \sqrt{n})$ rounds.

\begin{figure}[h]
	\begin{center}
		\includegraphics[width=0.4\textwidth]{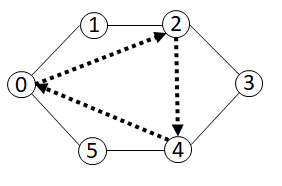}
	\end{center}
	\caption{\label{fig:barrier_reduction} Illustration of the reduction. The undotted edges are a $C_6$ cycle of $G$, and the dotted edges are a triangle of $G'$. We note that $G'$ has a second triangle on vertices $1,3,5$.}
\end{figure}

In Appendix~\ref{sec:app_c6_barrier}, for the sake of completeness, we follow the exact lines of \cite{EFFKO19} to show that for any $\alpha > 0$, showing a lower bound of $\Omega(n^\alpha)$ on directed triangle freeness implies strong circuit complexity lower bounds.

\section{Discussion: Intuition Regarding Round Complexity for $C_{2k}$ Detection in \congest}
\label{appendizx:c2kIntuition}

We believe that the $O(n^{1-1/k})$ round complexity for $C_{2k}$ detection in the \congest model is the best that can be achieved, barring some major improvement and a new approach which could have ramifications also for various other problems in the \clique model.
The reasons are as follows:

\subparagraph*{Listing 6-cycles in graphs with high conductance.}
As shown in \cite{CPZ19}, it takes $\Theta(n^{1-1/k+o(1)})$ rounds in the \congest model, \emph{in graphs with high conductance}, to perform subgraph \emph{listing} for $C_{2k}$. Therefore, primarily, if a faster algorithm for \emph{detection} is found for the general \congest model, it would imply a separation between detection and listing of $C_{2k}$ in the \congest model with high conductance. Such a separation could imply new algorithms for other problems related to subgraph listing also in the \clique model, due to the similarity between the \clique and \congest with high conductance models.

Further, notice it is highly likely that such an algorithm would function differently than any existing algorithm for $C_{2k}$ detection in the \clique model. The currently types of algorithms for $C_{2k}$ detection in the \clique model are split into several categories, as far as we are aware: (1) based on fast matrix multiplication, as in \cite{CHKKLPJ17}, and (2) based on sparsity aware listing, as seen in \cite{DLP12, CHKKLPJ17, CHLT20, PS18} and in this paper, and tend to leverage the T{\'u}ran number of $C_{2k}$ in order to list faster. These types of algorithms fail when moving to the \congest model in graphs with high conductance. The reason that the algorithms break is because in the \congest model, if the input graph is sparser, there is less bandwidth available to the entire network, in contrast with the \clique model. The first type of algorithms break since it is not known how to efficiently utilize input sparsity to improve the running time of fast matrix multiplication. This leads to the case where once the input graph is too sparse, the message complexity of the algorithm remains the same, and since the bandwidth available decreases, the round complexity increases. The latter type of algorithms break since while the T{\'u}ran number of $C_{2k}$ bounds the sparsity of the graph, and thus reduces the message complexity for the listing algorithms, it also implies less bandwidth for the network, effectively canceling out the effect of the reduction in message complexity.

\subparagraph{Finding 6-cycles in regular graphs.}
Since the T{\'u}ran number of 6-cycles is $\Theta(n^{4/3})$, when we restrict attention to \emph{regular} graphs, the ``interesting'' degree for $C_6$-detection 
is $\Theta(n^{1/3})$ (above this degree we know for sure that the graph contains a 6-cycle, and below this degree the problem becomes easier).

Suppose we assign to each node $u$ a random color $c(u)$ in $\set{0,\ldots,5}$ (i.e., we perform color-coding).
Each cycle is assigned consecutive colors $0,\ldots,5$ with constant probability, so we may as well search only for this type of cycle (and then repeat a constant number of times, to ensure that if there is a cycle, at least in one iteration, it will be colored correctly).
For a given node $u$ with $c(u) = 3$, let $L(u)$ be the 0-colored nodes that can be reached from $u$ by traversing a path of length 3 with descending colors $3,2,1,0$,
and let $R(u)$ be the 0-colored nodes that can be reached from $u$ by traversing a path of length 3 with ascending colors $3,4,5,0 = 6 \bmod 6$.
The problem of checking if $u$ participates in a well-colored 6-cycle boils down to checking whether $R(u) \cap L(u) = \emptyset$.

Since we are working with a regular graph of degree $\Theta(n^{1/3})$, in the ``average'' case (i.e., a random graph), we will have $|R(u)|, |L(u)| = \Theta(n)$. Checking whether two sets of size $\Theta(n)$ intersect or not is a famous problem in two-party communication complexity --- the $\textsf{Disjointness}$ problem, which is well-known to require $\Omega(n)$ bits of communication (or, more strongly, bits of \emph{information}) between the two players.
Intuitively, in order to check whether $L(u) \cap R(u) = \emptyset$, node $u$ must collect $\Theta(n)$ bits of information about each set, and since $u$ has degree $\Theta(n^{1/3})$, this requires $\Theta(n^{2/3})$ rounds.

Unfortunately, despite trying for a long time, we have not been able to make this intuition into a formal lower bound --- and now we see that at least there is a good reason for that, in the form of the barrier of Section~\ref{sec:barrier}.

\section{Acknowledgments}
This project was partially supported by the European Union's Horizon 2020 Research and Innovation Programme under grant agreement no.~755839, by the JSPS KAKENHI grants JP16H01705, JP19H04066 JP20H04139 and JP20H00579 and by the MEXT Q-LEAP grant JPMXS0120319794.

\bibliography{References}

\appendix
\section{Barrier of $\Omega(n^{\alpha})$ for Lower Bounds on Directed Triangle Freeness}
\label{sec:app_c6_barrier}

Following the exact lines in \cite{EFFKO19} of the barrier for triangle freeness, we show that for any $\alpha > 0$ showing a lower bound of $\Omega(n^\alpha)$ on directed triangle freeness implies strong circuit complexity lower bounds. This reduction is also very strongly based on the algorithm of \cite{CPZ19,CS19} for triangle enumeration.

The first step is to reduce from general graphs to graphs with high conductance:
let $\mathcal{A}_2$ be an algorithm that solves directed triangle freeness in a communication network with conductance $\phi$ and $n$ nodes, where every node is given as input $O(\deg(v))$ edges. Let $T_2(n,\phi)$ be the round complexity of $\mathcal{A}_2$. We show that given such an algorithm we can solve directed triangle freeness in $O(\mathcal{A}_2(n,\polylog{n})\log{n})$ rounds in \congest.

\begin{theorem}[Theorem 1 in~\cite{CS19}]
	\label{lem:expander_decompose}
	For $\epsilon \in (0,1)$, and a positive integer $k$, the network can partition its edges into two sets $E_r$,$E_m$ satisfying the following conditions:
	\begin{enumerate}
		\item The conductance of each connected component $G_i$ of $E_m$ satisfies $\Phi(G[V_i])\geq \phi$, where $\phi=(\epsilon/\polylog n)^{20 \cdot 3^k}$.
		\item  $|E_{r}| < \epsilon m$.
	\end{enumerate}
	This decomposition can be constructed using randomization in $O\left((\epsilon m)^{1/k}\cdot \left(\frac{\polylog n}{\epsilon} \right)^{20 \cdot 3^k}\right)$ rounds w.h.p.
\end{theorem}

We apply the theorem with $\epsilon = 1/6$, taking $k$ to be a large enough constant so that the round complexity of the decomposition round complexity is less than $O(n^\alpha)$. We call the set of vertices of every connected component of $E_m$ a \emph{cluster}. A node is called \emph{good} if it has more edges in $E_m$ than in $E_r$, and otherwise \emph{bad}. We call an edge $e \in E_m$ \emph{bad} if at least one of its endpoints is bad.

\begin{lemma}[\cite{CPZ19}]
	The number of bad edges is at most $2\epsilon m$.
\end{lemma}

Each cluster calculates its size $|U|$, and the number of edges in the cluster.
Since $\phi= \widetilde{O}(1)$, the diameter of each cluster is also $\widetilde{O}(1)$, and therefore this can be done in $\widetilde{O}(1)$ rounds (for example, by constructing a spanning tree on the cluster, and collecting the number of edges and nodes up the tree).
Then, each cluster runs $\mathcal{A}_2$ in parallel, where the input of each good node is all its edges (including edges leaving the cluster), and for each bad node, its edges in the cluster. We note that indeed by the definition of a good node, every node has $O(\deg_C(v))$ edges as input for $\mathcal{A}_2$, where $\deg_C(v)$ is $v$'s degree in its cluster. If $\mathcal{A}_2$ outputs that there exists a directed triangle, the cluster rejects and terminates. Otherwise, each cluster $U$ removes all good edges which are contained in $U$. The network then recurses on the remain edges until $O(1)$ edges remain. We note that as $|E_{good}| = m - |E_{bad}| -\epsilon m \geq m - 3\epsilon m = m/2$, in each iteration the network removes half of its edges, and the number of iterations are at most $O(\log{n})$. 

Clearly, if a node rejects then the graph contains a directed triangle. On the other hand, recall that the input graph of $\mathcal{A}_2$ is the edges adjacent to good nodes in $U$; therefore if $\mathcal{A}_2$ returns that there is no directed triangle, the network may safely remove all edges between two good nodes, as the triangle is contained in the union of inputs of both good endpoints.

The following lemma from~\cite{EFFKO19} shows that a dense cluster with good conductance is able to simulate a \emph{circuit},
where the size, depth, and input size of the circuit are related to the size, density and conductance of the cluster:
\begin{lemma}[\cite{EFFKO19}]
	\label{lem:circuit}
	Let $U$  be a graph $U=(V,E)$ with $|V|=n$ vertices and $m=N$ edges, with mixing time $\tau_{\textrm{mix}}$.
	Suppose that for some constant $c$,
	the function $f_N:\{0,1\}^{cN\log{n}}\rightarrow \{0,1\}$ is computed by a circuit $\mathcal{C}$ of depth $P$, consisting of gates with constant fan-in and fan-out,
	and at most $s\cdot N \cdot \log{n}$ wires for $s\leq n$. Then there is an $O(P \cdot s\cdot \tau_{\textrm{mix}}\cdot 2^{O(\sqrt{\log n}\log\log n)})$-round protocol in the \congest model on $U$ that computes $f_N$ in the network assuming the input is partitioned between the nodes such that each node has $O(\deg(v)\log{n})$ bits of input.
\end{lemma}

We consider the following family of functions $f_N:\set{0,1}^{2N\log{N}}\rightarrow\{0,1\}$: given an encoding of a graph\footnote{For a graph with $N$ edges $\{(u_i,v_i)\}_{i=1}^N$, the graph is encoded by the string $u_1.id,v_1.id,\dots,u_N.id,v_N.id$, where each id is padded to $\lceil \log{N} \rceil$ bits, and the rest of the string is padded in such a manner that indicates that there are no further edges.} with at most $N$ edges, does the graph contain a directed triangle?

\begin{corollary}
	If directed triangle freeness cannot be solved in less than $c_1n^{\alpha}$ rounds for any $c_1 > 0$, then there exist constants $c_2,c_3 > 0$ such that there is no  family of circuits that solve for all $N$ the function $f_N$ with $c_2N^{\alpha/4}/2^{c_3\sqrt{\log{n}}\log\log{n}}$ depth and at most $c_2N^{1+\alpha/4}/2^{c_3\sqrt{\log{n}}\log\log{n}}$ wires.
\end{corollary}
\begin{proof}
	Let $\mathcal{F}$ be an infinite family of graphs for which directed triangle freeness cannot be solved in less than $c_1n^{\alpha}$ rounds. Let $c_4 > 0$ be a constant such that the conductance of the clusters obtained by Theorem~\ref{lem:expander_decompose} is less than $\log^{c_4}{n}$. Assume by contradiction that for sufficiently large $c_2,c_3$ there is an infinite family of circuits solving for any $N$ the function $f_N$ with  $c_2N^{\alpha/4}/2^{c_3\sqrt{\log{n}}\log\log{n}}$ depth and $c_2N^{1+\alpha/4}/2^{c_3\sqrt{\log{n}}\log\log{n}}$ wires. Then by Lemma~\ref{lem:circuit} taking $s = n^{\alpha/2}/2^{c_3\sqrt{\log{n}}\log\log{n}}$ and  $P = n^{\alpha/2}/2^{c_3\sqrt{\log{n}}\log\log{n}}$ (both of which are larger than $c_2(n^2)^{\alpha/4}/2^{c_3\sqrt{\log{n}}\log\log{n}}$) there exists an algorithm with round complexity $c_1n^{\alpha}/\log{n}$ that solves directed triangle freeness on graphs with conductance at least $\phi = \log^{c_4}{n}$, where the input of each node is $O(\deg(v))$. By the reduction, we get that there is a $c_1 > 0$ such that directed triangle freeness can be solved in any network with $c_1 n^{\alpha}$ rounds, which is a contradiction.
\end{proof}

All together, we see that proving a lower bound of the form $\Omega(n^{\alpha})$ on directed triangle freeness, for any $\alpha > 0$,
would imply \emph{superlinear} lower bounds on the number of wires in circuits of \emph{polynomial depth}.
Such lower bounds on any explicit function are far beyond the reach of current circuit complexity techniques;
currently, the best lower bounds even for \emph{logarithmic}-depth circuits is at most linear in the input size (see, e.g. \cite{W19}).

\end{document}